\documentclass[a4paper,11pt,oneside,reqno]{amsart}   
\usepackage{a4wide}
\usepackage[top=1in, bottom=1.1in, left=1in, right=1in]{geometry}

\usepackage[foot]{amsaddr}
\usepackage{hyperref}
\usepackage{amssymb,amsmath,amscd,amsfonts,bbm,mathtools}
\usepackage{graphicx}
\graphicspath{{pic/} {newpic/} {newpic2/}}
\usepackage[boxed,vlined,longend]{algorithm2e}
\usepackage{algpseudocode}
\usepackage[usenames,dvipsnames]{xcolor}
\definecolor{violet}{rgb}{0.6,0.4,0.8}
\usepackage{enumerate,multicol,longtable,array}
\usepackage{wrapfig}
\usepackage[width=.75\textwidth]{caption}
\usepackage{kantlipsum}
\allowdisplaybreaks

\newcommand{\calD}{\mathcal{D}}
\newcommand{\calE}{\mathcal{E}}

\newcommand{\calG}{\mathcal{G}}

\newcommand{\calM}{\mathcal{M}}

\newcommand{\calP}{\mathcal{P}}
\newcommand{\calX}{\mathcal{X}}

\newcommand{\calZ}{\mathcal{Z}}

\newcommand{\frakm}{\mathfrak{m}}
\newcommand{\frakn}{\mathfrak{n}}

\newcommand{\calW}{\mathcal{W}}

\newcommand{\rmd}{\mathrm{d}}

\newcommand{\N}{\mathbb{N}}
\newcommand{\R}{\mathbb{R}}

\newcommand{\bbP}{\mathbb{P}}
\newcommand{\ep}{\varepsilon}

\def\l{\lambda}

\numberwithin{equation}{section}
\makeatletter
\providecommand*{\cupdot}{%
  \mathbin{%
    \mathpalette\@cupdot{}%
  }%
}
\newcommand*{\@cupdot}[2]{%
  \ooalign{%
    $\m@th#1\cup$\cr
    \sbox0{$#1\cup$}%
    \dimen@=\ht0 %
    \sbox0{$\m@th#1\cdot$}%
    \advance\dimen@ by -\ht0 %
    \dimen@=.5\dimen@
    \hidewidth\raise\dimen@\box0\hidewidth
  }%
}

\newtheorem{proposition}{Proposition}

\begin{document}
%

\title[A stochastic individual-based model for immunotherapy of cancer]{A stochastic individual-based model\\ for immunotherapy of cancer}

\author{Martina Baar$^{*}$}
\author{Loren Coquille$^{*}$}
\author{Hannah Mayer$^*$}

\author{Michael H\"olzel}

\author{Meri Rogava}
\author{Thomas T\"uting}

\author{Anton Bovier}

%

%
\date{Version of \today}

\begin{abstract} 
We propose an extension of a standard stochastic
individual-based model in population dynamics which broadens the range
of biological applications. Our primary motivation is modelling of 
immunotherapy of malignant tumours.  In this context  the different actors, T-cells, cytokines or cancer cells, are modelled as single particles  (individuals) in the stochastic system.
The main expansions of the model are distinguishing cancer cells by
phenotype and genotype,  including environment-dependent phenotypic plasticity that does not affect the genotype,  taking into account the effects of therapy and  introducing a competition term which lowers the reproduction rate of an
individual in addition to the usual term that increases
its death rate. 
 We illustrate the
new setup by using it to model various phenomena arising in immunotherapy. 
Our aim is twofold: on the one hand, we show that the interplay of
genetic mutations and phenotypic switches on different timescales as
well as the occurrence of metastability phenomena raise new
mathematical challenges.  
On the other hand, we argue why understanding purely stochastic events
(which cannot be obtained with deterministic models) may help to
understand the resistance of tumours to therapeutic approaches and
may have non-trivial consequences on tumour treatment protocols. This is supported through numerical simulations.
\end{abstract}

\keywords{adaptive dynamics, stochastic individual-based models,  cancer, im\-muno\-the\-rapy, \mbox{T-cells}, phenotypic plasticity, mutation.}
\subjclass[2000]{60K35, 92D25, 60J85.}

\thanks{
	\\We acknowledge financial support from the German Research Foundation (DFG) 
		through the
		\emph{Hausdorff Center for  Mathematics}, the Cluster of Excellence \emph{ImmunoSensation}, 
		the Priority Programme SPP1590 \emph{Probabilistic Structures in Evolution}, and the 
		Collaborative Research Centre 1060 \emph{The Mathematics of Emergent Effects}.
		We thank Boris Prochnau for programming. \\	
		\\
		\noindent MB, LC, HM, AB: Institute for Applied Mathematics, Bonn University\\
			\url{mbaar@uni-bonn.de},
			\url{lcoquill@uni-bonn.de},
			\url{hannah.mayer@uni-bonn.de},
			\url{bovier@uni-bonn.de}\\
		MH: Institute for Clinical Chemistry and Clinical Pharmacology, University Hospital, Bonn University\\
		\url{michael.hoelzel@ukb.uni-bonn.de} \\
			MR: Department of Dermatology, University Hospital, Magdeburg University\\
		\url{meri.rogava@ukb.uni-bonn.de}\\
		TT: 
		Department of Dermatology, University Hospital, Magdeburg University\\
		Laboratory of Experimental Dermatology,  University Hospital,	Bonn University\\
		\url{thomas.tueting@ukb.uni-bonn.de} \\
		\\$^*$ Equal contributions.
	}

\maketitle

\newpage
\setcounter{tocdepth}{2}
\tableofcontents

\section{Introduction}

Immunotherapy of cancer received a lot of attention in the medical as well 
as the mathematical modeling communities during the last decades \cite{Nowell:sf,Kuznetsov:1994fk,Eftimie:2011uq,hanahan2,GilVerGat,Holzel:2013ys}. Many different therapeutic approaches were 
developed and tested experimentally.
As for the classical therapies such as chemo- and radiotherapy,  \emph{resistance} is an important issue also for immunotherapy:
although a therapy leads to
an initial phase of remission, very often a relapse occurs. 
The main driving forces for resistance
are considered to be the genotypic and phenotypic heterogeneity of
tumors, which may be enhanced during therapy, see \cite{Holzel:2013ys,MarAlmPol,GilVerGat} and 
references therein. A tumor is a complex tissue which evolves in mutual influence with its environment~\cite{CorBis}. 

In this article, we consider the example of melanoma (tumors
associated to skin cancer) under T-cell therapy. Our work is
motivated by the experiments of Landsberg et al. \cite{Landsberg:2012vn}, which investigate
melanoma in mice under  \emph
{adoptive cell transfer} (ACT) therapy. This
therapeutic approach involves the injection of T-cells  which recognize a melanocyte-specific antigen and are able to kill differentiated types of melanoma cells. The therapy
induces an inflammation and the melanoma cells react to this environmental change by switching their
phenotype, i.e.\ by passing from a differentiated phenotype to
a dedifferentiated one (special markers on the cell surface disappear).
The T-cells recognize the cancerous cells through the markers which are
 downregulated in the dedifferentiated types. Thus, they are not capable of killing these cancer cells, 
 and a relapse is often observed. 
The phenotype switch is enhanced,
if pro-inflammatory cytokines, called TNF-$\alpha $ (Tumor Necrosis Factor), are
present. A second reason for the appearance of a relapse is that the T-cells become 
exhausted and are not working efficiently anymore. This problem was addressed by re-stimulation of the T-cells, but this 
led only to a delay in the occurrence of the  relapse. Of course, other immune cells and cytokines are also present. However,  according to the careful control experiments, their influence can be neglected in the context of the phenomena considered here. 

Cell division is not
required for switching, and switching is reversible. This means that the melanoma cells can recover their
initial (differentiated) phenotype \cite{Landsberg:2012vn}.  The switch is thus a purely phenotypic change which is not induced by a mutation.   
The state of the
tumor is a mixture of differentiated and
dedifferentiated cells. 
One possibility to avoid a relapse is to inject two
 types of T-cells (one specific to
differentiated cells as above, and the other
specific to dedifferentiated cells) as suggested also in \cite{Landsberg:2012vn}.  

In this paper, we propose a quantitative mathematical model that can reproduce the phenomena observed  in the experiments of \cite{Landsberg:2012vn}, 
and which allows to 
simulate different therapy protocols, including some where several types of T-cells are used. 
The model we propose is an extension of the \emph{individual-based} stochastic models for adaptive dynamics that were introduced in  
Metz et al.\ \cite{metz-geritz-al-96} and developed and analyzed by many authors in recent years (see e.g. \cite{BolPac1, 
BolPac2, DieLaw,champagnat-ferriere-al-01,Cha2006,ChaFerMel,BovWan2013,ChaJabMel,CHLM}) to the setting of tumor 
growth under immunotherapy. Such models are also referred to as \emph{agent-based} models or \emph{particle systems}. Note that we use the term individual for T-cells, cytokines or cancer cells, in particular not for patients or mice.

These stochastic models describe the evolution of interacting cell populations, in which the relevant events for each individual (e.g.\ birth and death) occur randomly.  

It is well known that in the limit of large cell-populations, these models are approximated by deterministic kinetic rate models, 
which are widely used in the modeling of cell populations. 
However, these approximations are inaccurate and fail to account for important phenomena, if the numbers of some 
sub-populations become small. In such situations, random fluctuation may become highly significant and completely alter the 
long-term behavior of the system.
For example,  in a phase of remission during therapy, the cancer and the T-cell populations drop to a
low level and may die out due to fluctuations.

A number of (mostly deterministic) models have been proposed that describe the development of a tumor
under treatment, focusing on different aspects. For example, a deterministic model for ACT therapy is presented  
in \cite{Eftimie:2011uq}. Stochastic approaches were
used to understand certain aspects of tumor development, for
example rate models \cite{Gupta} or multi-type branching
processes; see the book by Durrett \cite{Durrett15} or  \cite{Boz,AntKrap,Durrett13}.  To our knowledge, however, it is a novel 
feature of our models to describe the coevolution of immune- and tumor cells 
taking into account both interactions and phenotypic plasticity. 
Our models can help understanding the interplay of therapy and resistance, in particular in the case of immunotherapy,
and may be used to predict successful therapy protocols.
The main expansions of our models are the following: 

\begin{itemize}
\item Two types of transitions are allowed: genotypic mutations and phenotypic switches. The characteristic timescale for a mutation to occur can be significantly longer than a timescale for epigenetic transitions. 
\item Phenotypic changes may be affected by the environment which is not modeled deterministically as in \cite{ChaJabMel} but as particles undergoing the random dynamics as well.
\item A predator-prey mechanism is included (modeling the interaction of cancer cells and immune cells). The defense strategies of the prey are not modeled by different interaction rates as in \cite{CHLM} but by actually modeling the escape strategy, namely switching.
\item A birth-reducing competition term is included which takes account of the fact that competition between individuals may also affect their reproduction behavior.
\end{itemize}

The paper is structured as follows.
In Section \ref{sec-model} we define the model and state the convergence towards a quadratic system of ODEs in the large population limit.
In Section \ref{sec-relapse}, we first present an example which qualitatively models the therapy carried out in
Landsberg et al.\ \cite{Landsberg:2012vn} (Subsection \ref{sec-1tcell}). We point out a phenomenon of relapse caused by random fluctuations:
in the phase of remission, the typical number of T-cells is small, and random fluctuations can cause their extinction, allowing for a growth of the tumor. 
As a second example (Subsection \ref{sec-2tcells}) we study the T-cell therapy with two types of
T-cells. In this case an ever richer class of possible behavior
occurs: either the tumor is cured, or one type of T-cells dies out, or both types of T-cells do, or all 
populations survive for a given time. 
Subsection \ref{sec-bio-param} presents our choice of physiologically reasonable parameters, and the reproduction of experimental observations together with predictions.
In Section \ref{sec-mutant} we consider the case of rare mutations.
Without therapy (Subsection \ref{sec-PES}), we first show how to extract an effective Markov chain on the space of
genotypes, whose transition rates are given by the asymptotic behavior
of a faster Markov process on the space of phenotypes. In particular, we define a notion of invasion fitness in this setting, and describe how we can obtain a generalization of the Polymorphic Evolution Sequence introduced in \cite{ChaMel2011}. 
Second (Subsection \ref{sec-brc}), we study the interplay of mutation and therapy. We consider birth-reducing competition
between tumor cells and show that the
appearance of a mutant genotype may be enhanced under
treatment.

\newpage

\section{The model}\label{sec-model}

We introduce a general model, which contains three types of actors:
\begin {itemize}
\item \textit {Cancer cells:} each cell is characterized by a genotype and a
phenotype. These cells can divide (with or without mutation), die (due to age, competition or therapy)  and switch their phenotype.
We assume that the switch is inherited by the descendants of the switched cells.
\item \textit {T-cells:} each cell can  divide, die and produce
cytokines.
\item \textit {Cytokines:} each messenger can vanish and influence the switching of cancer cells.
\end {itemize}

The trait space, $\calX$, is a finite set of the form
\begin{align}
&\calX=\calG\times\calP\cupdot\calZ\cupdot\calW 
=\left\{g_1,\ldots,g_{|\calG|}\right\}\times\left\{p_1,\ldots,p_{|\calP|}\right\}\nonumber \\ 
&\hspace{0.7cm}\cupdot \left\{z_1,\ldots,z_{|\calZ|}\right\}
\cupdot \left\{w_1,\ldots,w_{|\calW|}\right\} 
\end{align}
where $(g,p)\in \calG \times \calP$ denotes a cancer cell with genotype $g$
and phenotype $p$, $z\in\calZ$ a T-cell of type $z$ and $w\in\calW$
a cytokine of type $w$. We write $|\cdot|$ for the number of elements of a set and $\cupdot$ for 
disjoint unions of sets.

\subsection{General notations and parameters}\label{notations}

A \emph{population} at time $t\in \R_+$ is represented by  the measure 
\begin{equation}
\nu_t^K=\frac 1K \sum_{x\in \calX} \nu_t(x) \delta_x,
\end{equation}
where $\nu_t(x)$ is the number of individuals of type $x$ at time $t$ and $\delta_x$ denotes the Dirac measure 
at $x$. Here, $K$ is a parameter  that allows to scale the population size and is usually called \emph{carrying-capacity}
		of the environment. The dynamics of the population is described by a continuous time Markov process, $(\nu_t^K)_{t\geq0}$, with the following rates:

\vspace{0.2cm}
\noindent  Each cancer cell of type $(g,p)$ is characterized by
\begin{itemize}
	\item natural birth and death rates:  $b(p)\in\R_+$ and $d(p)\in\R_+$. 
	\item competition kernels: $c(p,\tilde p)K^{-1}\!\in\!\R_+$ and  \mbox{$c_b(p,\tilde p)K^{-1}\!\in\!\R_+$}, where the first term 			increases the death rate and the second term, called birth-reducing
	competition, lowers the birth rate of a cancer cell of phenotype $p$ in
	presence of a cancer cell of phenotype $\tilde p$.   
	If the total birth rate is already at a level 0, then $c_b(p,\tilde p)K^{-1}\in \R_+$ acts as an additional death rate.
	\item therapy kernel: $t(z,p)K^{-1}\in\R_+$ additional death rate of
	a cancer cell of phenotype $p$ due to the presence of a T-cell of
	type $z$. In addition,  $\ell^{\text{\tiny kill}}_w(z,p)\in\N_0$ cytokines of type $w$ are deterministically produced at each  killing event.
	\item switch kernels: $s^g(p,\tilde p)\in\R_+$ and $s_w^g(p,\tilde p)K^{-1}\in\R_+$ denote the natural and 
	cytokine-induced switch kernels from a cancer cell of type
	$(g,p)$ to one of type $( g,\tilde p)$. 
	\item mutation probability and law: $\mu_g\in[0,1]$  denotes the probability that a birth event of a cancer cell of 			genotype $g$ is a mutation.
	$m((g,p),(\tilde g,\tilde p))\in[0,1]$ encodes the
	probability that, whenever a mutation occurs, a cancer cell of
	type $(g,p)$ gives birth to a cancer cell of type $(\tilde g,\tilde
	p)$. By definition $m((g,p),( g,p))=0$  and $\sum_{\tilde g, \tilde p} m((g,p),(\tilde g,\tilde p))=1$.
\end{itemize}
Each T-cell of type $z$  is characterized by
\begin{itemize}
	\item natural birth and death rates: $b(z) \in\R_+$  and $d(z) \in\R_+$. 
	\item reproduction kernel: $b(z,p)K^{-1} \in\R_+$ denotes the rate of reproduction of a T-cell with trait $z$
	in presence of a cancer cell of phenotype $p$. In addition, $\ell^{\text{\tiny prod}}_w(z,p) \in\N$ cytokines of type $w$ are deterministically produced at each reproduction event.
\end{itemize}
Each cytokine of type $w$ is characterized by 
\begin{itemize}
	\item  natural death rate: $d(w) \in\R_+$.
	\item  the molecules are produced when a cancer cell dies due to therapy or a T-cell reproduces.
\end{itemize}

Note that the relation between $\calG$ and $\calP$ is encoded in the switch kernels. They specify which phenotypes are expressed by a given genotype and influence the proportions of the different phenotypes in a (dynamic) environment. 

Figure \ref{fig-rates} provides a graphical representation of the transitions for a population with trait space $\calX= \{x=(g,p),y=(g,p')\}\cupdot \{ z_x\}\cupdot \{w\}$, which constitute our model for the ACT therapy described in Landsberg et al.\ \cite{Landsberg:2012vn}, see Subsection \ref{sec-1tcell}.

\begin{figure}[h!]
	\centering
	\begin{minipage}{.6\textwidth} 
		\includegraphics[width=\textwidth]{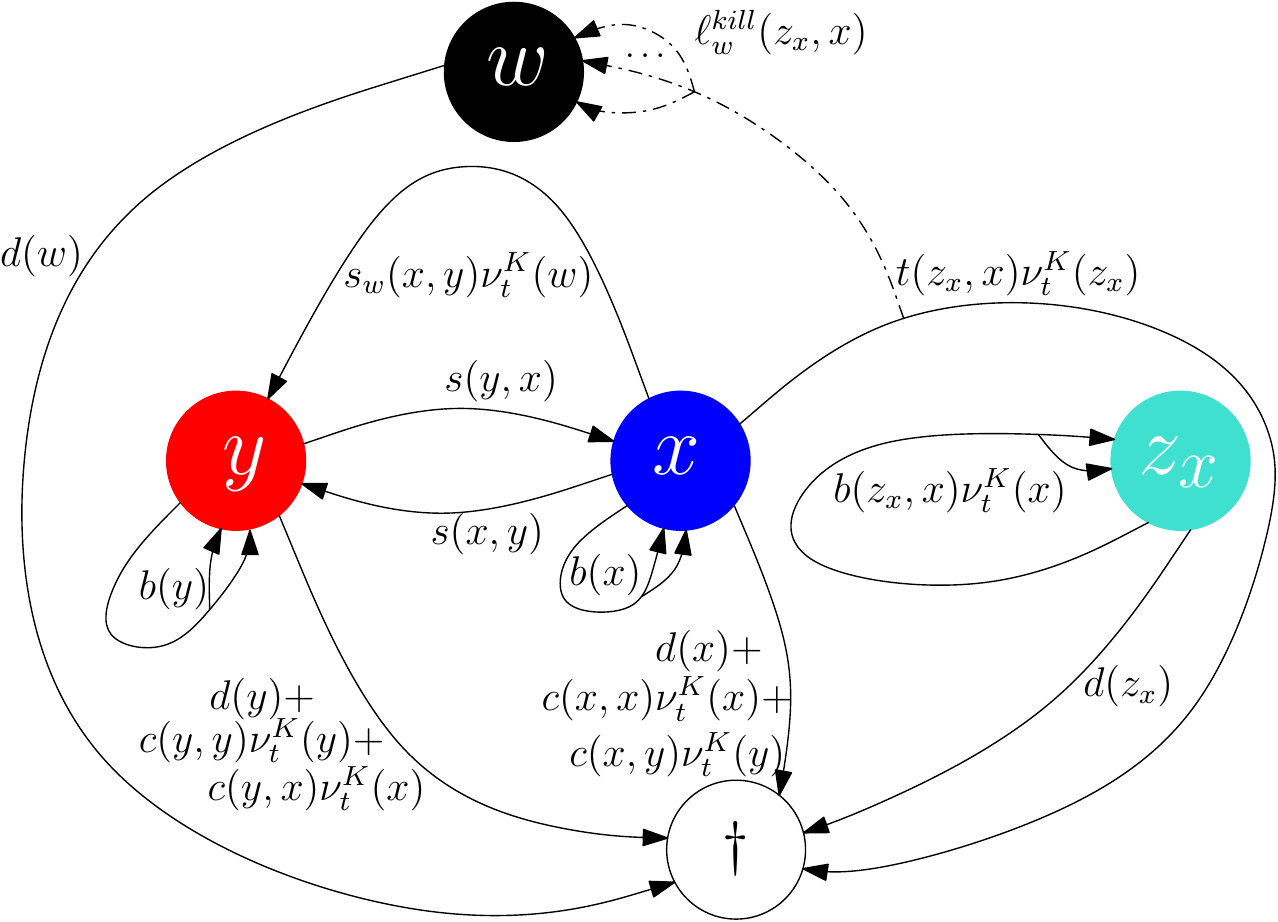} 
	\end{minipage}
	\caption{\small {Dynamics of the process (without mutations) modelling the experiments described in \cite{Landsberg:2012vn}. Here, $x$ denotes differentiated melanoma cells, $y$ dedifferentiated melanoma cells, $z_x$ T-cells and $w$ TNF-$\alpha$. At each arrow the rate for occurrence of the corresponding event is indicated (e.g.\ birth is illustrated with two arrowheads and death with an arrow directed to $\dagger$).}\label{fig-rates}}
\end{figure}

\subsection{The dynamics}

Let  $\mathcal M^K(\calX)\equiv\left\{\tfrac 1K \sum_{i=1}^n \delta_{x_i}: n\in\N, x_1,\ldots, x_n\in\calX\right\}$, i.e. the set of finite point measures on $\calX$ rescaled by $K$.
Then, for each $K\in\N$, the dynamics of the $\mathcal M^K(\calX)$-valued Markov process,
${(\nu^K_t)}_{t\geq0}$, describing the evolution of 
the population at each time $t$, can be summarized as follows: 

At time $t=0$ the population is characterized by a given measure $\nu^K_0 \in \mathcal M^K(\calX)$.
Each individual present at time $t$ has several exponential clocks with intensities depending on its trait $x\in\mathcal X$  and on the current state of the system. We use the shorthand notation $\nu_t^K(p)\!=\!\sum_{g\in\calG}\nu_t^K(g,p)$ 
and  $\lfloor\cdot \rfloor_\pm$ to denote the positive/negative part of the argument.
\vspace{0.2cm}
\begin{enumerate}
	\item
	Each present cancer cell of trait $(g,p)\in\calG\times\calP$ has \\[0.5em]
	-\: a \emph{clonal reproduction} clock with rate  
	\quad $ (1-\mu_g)\,\left\lfloor b(p)-\sum_{\tilde p\in\calP}c_b(p,\tilde p)\nu^K_t{(\tilde p)}\right\rfloor_+$.\\[0.5em]
	Whenever this one rings, an additional cancer cell of the same trait
	$(g,p)$ appears.\\[0.5em]
	-\: a \emph{mutant reproduction} clock with rate \quad
	$\mu_g\,\left\lfloor b(p)-\sum_{\tilde p\in\calP}c_b(p,\tilde  p)\nu^K_t{(\tilde p)}\right\rfloor_+$.\\[0.5em]
	Whenever  this one rings, a cancer cell of trait $(\tilde
	g,\tilde p)$ appears according to the kernel $m((g,p),(\tilde g,\tilde p))$. \\[0.5em]
	-\: a \emph{natural mortality} clock with rate $ d(p)+\sum_{\tilde p\in\calP} c(\tilde p,p) \nu^K_t{ (\tilde p)}
	+\left\lfloor b(p)-\sum_{\tilde p\in\calP}c_b(p,\tilde p)\nu^K_t{(\tilde p)}\right\rfloor_-$.\\[0.5em]
	Whenever  this one rings, the cancer cell disappears.\\[0.5em]
	-\: a \emph{therapy mortality} clock of rate
	\quad$\sum_{z\in\calZ}t(z,p) \nu^K_t(z)$.
	\\[0.5em]Whenever  this one  rings, the cancer cell disappears and  $\ell^{\text{\tiny kill}}_w(z,p)$ cytokines of type $w$ appear according to the weights $t(z,p) \nu^K_t(z)$.\\[0.5em]
	-\: a \emph{natural and cytokine-induced switch} clock with rate 
	$\sum_{\tilde p\in\calP} (s^g(p,\tilde p) + \sum_{w\in\calW}s^g_w(p,\tilde p) \nu^K_t(w))$. \\[0.5em]
	Whenever  this one rings, this cancer cell  disappears and a new cancer cell of trait
	$(g,\tilde p)$ appears according to the weights $s^g(p,\tilde p) + \sum_{w\in\calW}s^g_w(p,\tilde p) \nu^K_t(w)$.\\
	\item Each present T-cell of trait $z\in\calZ$ has \\[0.5em]
	\begin{tabular}{ll}
		-\;  a \emph{natural birth} clock with rate &$b(z) $\\[0.5em]
		-\;  a \emph{natural mortality} clock with rate \qquad   \qquad&$d(z)$\\[0.5em]
		-\; a \emph{reproduction} clock with rate&$\sum_{p\in\calP} b(z,p)\nu^K_t{(p)}.$
	\end{tabular}\\[0.5em]
	Whenever the reproduction clock  of a T-cell rings, a additional T-cell with the same trait
	$z$ and $\ell^{\text{prod}}_w(z,p)$ cytokines of type $w$ appear according to the weights $ b(z,p)\nu^K_t{(p)}$. \\
	\item Each present cytokine has \\[0.5em]
	\begin{tabular}{ll}
		\;-\; a \emph{mortality} clock
		with rate \qquad \qquad&$d(w).$
	\end{tabular}\\[0.5em]
	Moreover, an amount of $\ell^{\text{kill}}_w(z,p)$ particles (of trait $w$) are
	produced every time a T-cell of trait $z$ kills a cancer cell of
	phenotype $p$, and a number of $\ell^{\text{prod}}_w(z,p)$ particles are produced
	every time
	a T-cell of trait $z$ is produced in the presence of a cancer cell of phenotype $p$.\\
\end{enumerate}
\vspace{-0.5cm}
The measure-valued process $(\nu_t^K)_{t\geq0}$ is a Markov process whose law is
characterized by its infinitesimal generator $L^K$ which captures the
dynamics described above (cf. \cite{EthKur} Chapter 11 and \cite{FouMel2004}).
 The generator acts on bounded measurable functions $\phi$ from $\calM^K$ into $\R$, for all $\eta\in\calM^K$ by 
\vspace{-0.5cm}
\begin{align}\label{BRC}\nonumber
\left(L^K\phi\right)(\eta)
\:&=\:
\sum_{(g,p)\in\calG\times \calP}\left(\phi\left(\eta+\tfrac{\delta_{(g,p)}}K\right)-\phi(\eta)\right)(1-\mu_g)\Bigg\lfloor\,b(p)-\sum_{\tilde p\in\calP}c_b(p,\tilde p)
\eta(\tilde p)\Bigg\rfloor_+K\eta(g,p)\\[0.2em]\nonumber
&+ \sum_{(g,p)\in\calG\times \calP}\left(\phi\left(\eta-\tfrac{\delta_{(g,p)}}K\right)-\phi(\eta)\right)\\\nonumber
& \hspace{2cm}\times \Biggl(d(p)+\sum_{\tilde p\in\calP} c(p,\tilde p)\eta (\tilde p) 
+\Bigg\lfloor b(p)-\sum_{\tilde p\in\calP}c_b(p,\tilde p)\eta(\tilde p)\Bigg\rfloor_- \Biggr)K\eta(g,p)\\[0.2em]\nonumber
&+ \sum_{(g,p)\in\calG\times \calP}\sum_{z\in\calZ}
\left(\phi\left(\eta-\tfrac{\delta_{(g,p)}}K+{\textstyle\sum_{w\in\calW}}\:\ell^{\text{kill}}_w(z,p)\tfrac{\delta_w}K\right)-\phi(\eta)\right)
t(z,p)\eta(z) K\eta (g,p)\\[0.2em]\nonumber
&+ \sum_{(g,p)\in\calG\times \calP}\:\sum_{\tilde p\in\calP}
\left(\phi\left(\eta+\tfrac{\delta_{(g,\tilde p)}}K-\tfrac{\delta_{(g,p)}}K\right)-\phi(\eta)\right)\\\nonumber 
&\hspace{2.7cm} \times \left(\vphantom{\sum}s^g(p,\tilde
p)+\textstyle{\sum_{w\in\calW}}
s^g_w(p,\tilde p)\eta( w)\vphantom{\int}\right)K\eta
(g,p)\\[0.2em]\nonumber
&+\;\sum_{z\in\calZ}\;\sum_{p\in\calP}\left(\phi\left(\eta+\tfrac{\delta_{z}}K+{\textstyle\sum_{w\in\calW}} \ell^{\text{prod}}_w(z,p)\tfrac{\delta^{}_w}K\right)-\phi(\eta)\right)\Bigg(b(z,p)\eta(p)\Bigg)K\eta(z)\\\nonumber 
&+\;\sum_{z\in\calZ}\left(\phi\left(\eta+\tfrac{\delta_{z}}K\right)-\phi(\eta)\right)b(z)K\eta(z)
+\;\sum_{z\in\calZ}\left(\phi\left(\eta-\tfrac{\delta_{z}}K\right)-\phi(\eta)\right)d(z)K\eta(z)\\\nonumber
&+\;{\sum_{w\in\calW}}\left(\phi\left(\eta-\tfrac{\delta_{w}}K\right)-\phi(\eta)\right)d(w)K\eta(w)\\\nonumber
&+ \sum_{(\tilde g, \tilde p)\in\calG\times \calP} \sum_{(g,p)\in\calG\times \calP}
\left(\phi\left(\eta+\tfrac{\delta_{(\tilde g,
		\tilde  p)}}K\right)-\phi(\eta)\right)\\
 & \hspace{2cm}
\times \mu_g m((g,p),(\tilde g,\tilde p))\Bigg\lfloor b(p)-\sum_{ p'\in\calP}c_b(p, p')\eta( p')\Bigg\rfloor_+ K\eta(g,p).
\end{align}

All the figures of this article containing simulations of the above model were generated
by a computer program written by Boris Prochnau. 
The pseudo-code can be found in the
Appendix~\ref{appendix}.

\subsection{Large population approximation}

There is a natural way to jointly construct the processes $(\nu_t^K)_{t\geq0}$ for different values of $K$, see \cite{EthKur}.
For this joint realization, the sequence of rescaled processes $((\nu_t^K)_{t\geq0})_K$ converges almost surely as $K$ tends to infinity to 
the solution of a quadratic system of ODEs, as stated in the following proposition, which can be seen as a law of large numbers.
The deterministic system, which provides (partial) information on the stochastic system, consists of a logistic part, a predator-prey relation between T-cells and cancer cells, a mutation and a switch part.


Note that the population can be represented
as a vector $V_K(t):=(\nu^K_t(x))_{x\in \calX}$ of dimension  $d=
|\calG|\cdot |\calP|+|\calZ|+{|\calW|}$. 

\begin{proposition}\label{det-limit}
	Suppose that the initial conditions converge almost surely to a deterministic
	limit, i.e.\ $\lim_{K\to\infty} V_K(0)=v(0)$. Then, for each $T\in \R_+$ the sequence of rescaled processes $\left(V_K(t)\right)_{0\leq t\leq T}$ 
	converges almost surely as
	$K\to\infty$ 
	to the $d$-dimensional deterministic process which is the
	unique solution to the following quadratic dynamical system: 
	
\noindent For all $(g,p) \in \calG\times\calP$,

\begin{align}\label{SD1}
\frac{\rmd \mathfrak n_{(g,p)}}{\rmd
  t}\nonumber\:=&\;
\frakn_{(g,p)}\:\Bigg((1-\mu_g)\,\Bigg\lfloor b(p)-\sum_{(\tilde g, \tilde p)\in\calG\times\calP}c_b(p,\tilde p)\frakn_{(\tilde g,\tilde p)}\Bigg\rfloor_+\\\nonumber&\hspace{1.5cm} -\Bigg\lfloor b(p)-\sum_{(\tilde g, \tilde p)\in\calG\times\calP}c_b(p,\tilde p)\frakn_{(\tilde g,\tilde p)}\Bigg\rfloor_-
-d(p) -\sum_{(\tilde g, \tilde p)\in\calG\times\calP}  c(p,\tilde p) \frakn_{ (\tilde g,\tilde p)}\nonumber\\ &
\hspace{1.5cm}- \sum_{z\in\calZ} t(z,p) \frakn_z -\sum_{\tilde p\in\calP} \left(s^g(p,\tilde p)
+{\sum_{w\in\calW}}
s_w^g(p,\tilde p) \frakn_w\right)\Bigg) \nonumber\\\nonumber&
+\sum_{\tilde p\in\calP}  \;\frakn_{(g,\tilde p)}\:
\Bigg(s^g(\tilde p,p)+
{\sum_{w\in\calW}}s^g_w(\tilde p,p) \frakn_w\Bigg)\nonumber\\
&+\sum_{(\tilde g, \tilde p)\in\calG\times\calP}\frakn_{(\tilde g,\tilde p)}
\Bigg(\mu_{\tilde g} \,m((\tilde g,\tilde p),( g,
p))\Bigg\lfloor b(\tilde p)-\sum_{ (g',p')\in\calG\times\calP}c_b(\tilde p, p') \frakn_{( g', p')}\vphantom{\int}
\Bigg\rfloor_+\Bigg),\nonumber
\end{align}
for all  $\quad z\in \calZ$, 
\begin{align}
\frac{\rmd \frakn_z}{\rmd t}\:=&\;
\frakn_{z} \:\Bigg(b(z)-d(z)+ \sum_{(g,p)\in\calG\times \calP}b(z,p) \frakn_{(g,p)}\Bigg), 
\end{align}
for all  $w\in \calW$,
\begin{align}
\frac{\rmd \frakn_w}{\rmd t}\nonumber\:=&\;
 -\frakn_w d(w) + \sum_{(g,p)\in\calG\times \calP}\frakn_{(g,p)}
\sum_{z\in\calZ} \Big(\ell^{\textup{kill}}_w(z,p) \,
    t(z,p)+\ell^{\textup{prod}}_w(z,p)\,b(z,p)\Big) \frakn_z .
\end{align}
	More precisely, $\mathbb P\left(\lim_{K\to\infty}\sup_{0\leq t \leq T}|V_K(t)-\frakn(t)|=0\right)=1$, where $\frakn(t)$ denotes the solution to Equations \eqref{SD1} with initial condition $v(0)$.
\end{proposition}

\begin{proof}
This result follows from 
Theorem 2.1 in Chapter 11 of \cite{EthKur}. 
\end{proof}

It is an important feature of  stochastic models opposed to deterministic ones that populations can die out. There are two main reasons for the extinction of a population for finite~$K$: 
first, the trajectory of the population size is transient and passes typically through a low minimum. In this case, random fluctuations can lead to extinction before the population reaches its equilibrium.
Second, fluctuations around a finite equilibrium cause extinction of a population after a long enough time.
The second case happens at much longer times scales than the first one. 
In both cases, the value of $K$ plays a crucial role, since it determines the amplitude of the fluctuations and thus the probability of extinction. 

The relevant mutations in the setup of cancer therapy are driver mutations and appear only rarely. In this case, more precisely, when  the mutation probabilities, $\mu_g\equiv\mu_g^K$, tend to zero as $K\to\infty$, mutations are invisible in the deterministic limit. Due to the presence of the switches the analysis of the system is difficult. It is \emph{not} a generalized Lotka-Volterra system of the form $\dot{\frakn}=  \frakn f( \frakn)$, where $f$ is linear in $\frakn$. 
In Section \ref{sec-mutant} we show how to deal mathematically with rare mutations and their interactions with fast phenotypic switches or therapy.

\section{Relapse due to random fluctuations}\label{sec-relapse}
\subsection{Therapy with T-cells with one specificity}\label{sec-1tcell}

In this section we present an example which \textit{qualitatively} models the experiment of Landsberg et al.\ \cite{Landsberg:2012vn}, where melanoma escape ACT therapy by phenotypic plasticity in presence of
TNF-$\alpha$.
Mutations are not considered here, since this was not investigated in the
experiments. 

Let us denote by $ x:=(g,p) $ the 
differentiated cancer cells,
by $ y: =(g,p')$ the dedifferentiated cancer cells, by $z_x$ the T-cells attacking only cells of type $x$, and by $ w $ the TNF-$\alpha $ proteins. 
We start with describing the deterministic system and denote by $\frakn_i$  its solution for trait $i\in \{x,y,z_x,w\}$. 
It is the solution to the following system of four differential equations:
\begin{align}\label{1tcell-system}
\nonumber
\dot \frakn_{x} &=\frakn_{x}
\big(b({x})-d({x})-c({x},{x}) \frakn_{x}\!-c({x},y) \frakn_y\! -s(x,y)-s_w(x,y) \frakn_w\! -t({z_x},x)
\frakn_{z_x}\big) 
+s(y,x) \frakn_y \\\nonumber
\dot \frakn_y &= \frakn_y \big(b(y)-d(y)-c({y,y}) \frakn_y-c({y,x}) \frakn_{x}-s({y,x})\big)+s(x,y)\frakn_x+s_w(x,y) \frakn_w \frakn_{x} \\\nonumber
\dot \frakn_{z_x} \!&= \frakn_{z_x}(b({z_x},x) \frakn_{x}-d({z_x})) \\
\dot \frakn_w &= -\frakn_w d(w) 
+(\ell^{\text{kill}}_w({z_x},x) \,t({z_x},x) 
+\ell^{\text{prod}}_w({z_x},x) \,b({z_x},x) )\frakn_{x}\frakn_{z_x}.
\end{align}
We fix the following parameters:
\begin{equation}\label{1tcell-muster-parameter}
\begin{array}{l@{\hspace{1cm}}l@{\hspace{1cm}}l@{\hspace{1cm}}l}
b({x})=3 			& b(y)=3		& b({z_x},{x})=8	 	&\ell^{\text{kill}}_w({z_x},x)=1\\
d({x})=1 			&	 d(y)=1		& t({z_x},{x})=28		& \ell^{\text{prod}}_w({z_x},x)=0\\
c({x},{x})=0.3 		&	c(y,{x})=0 	& d({z_x})=3			& d(w)=15\\
c({x},y)=0			&	c(y,y)=0.3	&						& s_w(x,y)=4			\\
s(x,y)=0.1			&	s(y,x)=1	& 						&	
\end{array}
\end{equation}
and  initial conditions:
\begin{equation}\label{1tcell-muster-incond}
\begin{array}{l}
 (\frakn_{x},\frakn_{y},\frakn_{z_x},\frakn_{w})(0)=(2,0,0,0.05,0).
\end{array}
\end{equation}
The solution to the deterministic system \eqref{1tcell-system} with parameters \eqref{1tcell-muster-parameter} and initial conditions \eqref{1tcell-muster-incond} can be seen on Figure \ref{det-syst} (A).
There are three fixed points in this example, see the black dots on Figure \ref{fig-fixedpoints} (B):
$P_{0000}$ where all populations sizes are zero, 
$P_{xy00}$ where the T-cells and TNF-$\alpha$ are absent  and both melanoma populations are present and
$P_{xyz_xw}$ where all populations are present.
$P_{xyz_xw}$ is the only stable fixed point and $P_{xy00}$ is stable in the invariant subspace $\{\frakn_{z_x}=0\}$ (i.e.\ if the T-cell population is zero).
To highlight the qualitative features of the system, we choose parameters such that the minimum of the T-cell population during remission is low, and such that the equilibrium value of melanoma of type $x$ in presence of T-cells is low, whereas equilibrium values of both melanoma types in absence of T-cells are high.

For initial conditions such that the number of differentiated melanoma cells, $\frakn_{x}(0)$, is large, the number
of injected T-cells, $\frakn_{z_x}(0)$, is small, and the numbers of dedifferentiated melanoma  cells, $
\frakn_y(0)$,  and TNF-$\alpha$ molecules, $\frakn_w(0) $, are small or equal to zero, the \emph{deterministic} system is attracted to $P_{xyz_xw}$: the T-cell population, $\frakn_{z_x}$, increases in presence of its target $x$, TNF-$\alpha$ is secreted, and the
population of differentiated melanoma cells, $\frakn_x$, shrinks due to killing and TNF-$\alpha$ induced switching, whereas the population of dedifferentiated melanoma cells,  $\frakn_y$, grows. 

For the stochastic system, several types of behavior can occur with
certain probabilities: either the trajectory stays close to that of the 
deterministic system and the system reaches a neighborhood of the fixed point $P_{xyz_xw}$ (Fig. \ref{fig-1tc} (A)) or the T-cell population, $\nu^K(z_x)$, dies out and the system reaches a neighborhood of  $P_{xy00}$ (Figure \ref{fig-1tc} (B)). In the latter case  the TNF-$\alpha$ population, $\nu^K(w)$,
becomes extinct shortly after the extinction of the T-cells, $\nu^K(z_x)$, and the population of differentiated melanoma cells, $\nu^K(x)$, can grow again.  Moreover, TNF-$\alpha$ inducing the switch from ${x}$ to $y$ vanishes and we observe
a relapse which consists mainly of differentiated cells.
Depending on the choice of parameters (in particular switching, therapy or cross-competition), a variety of different 
behavior is possible.

\begin{figure}[h!]
	\centering
	\begin{minipage}{.49\textwidth} 
		A \includegraphics[width=\textwidth]{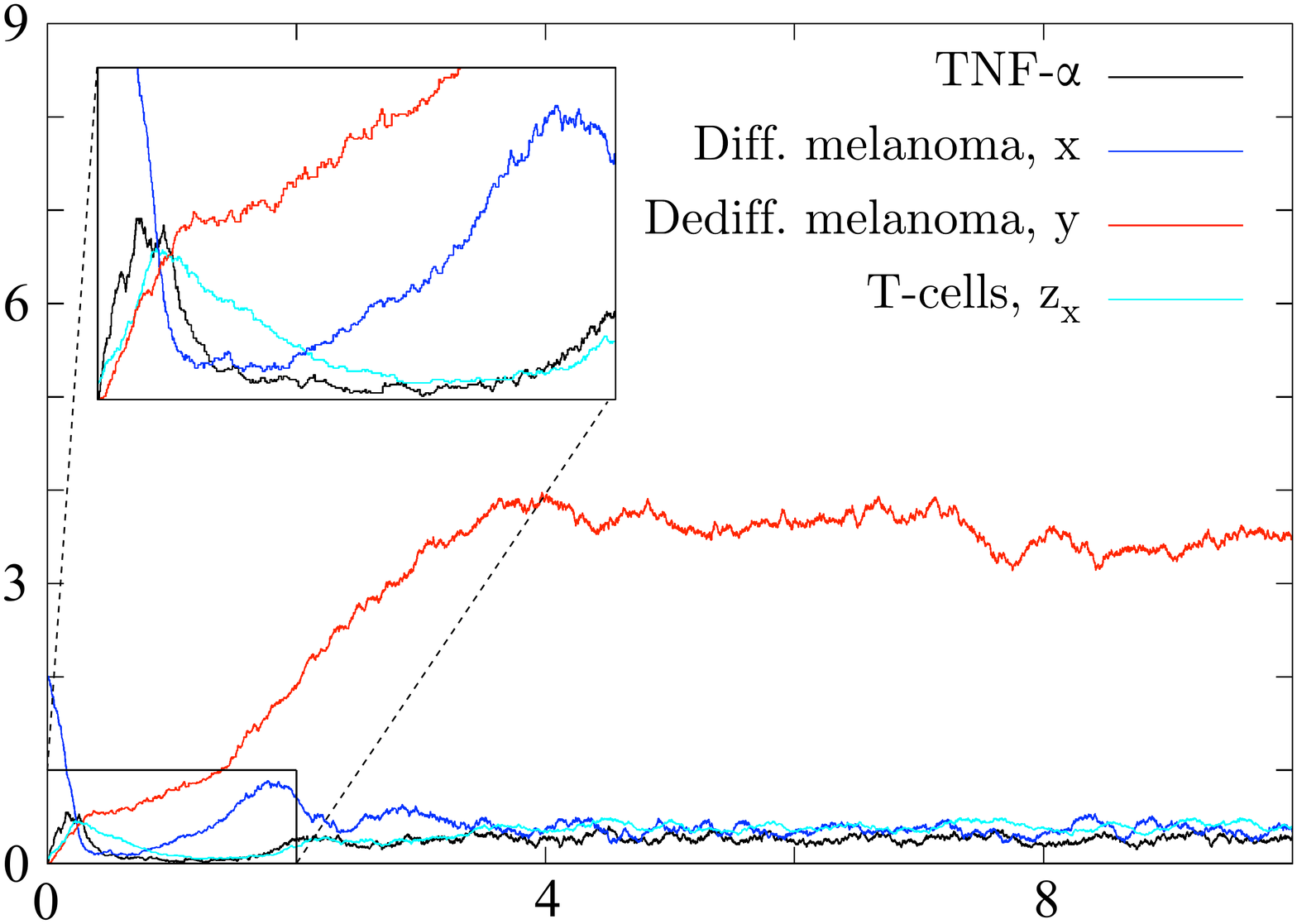}
	\end{minipage}
	\hfill
	\begin{minipage}{.49\textwidth} 
		B \includegraphics[width=\textwidth]{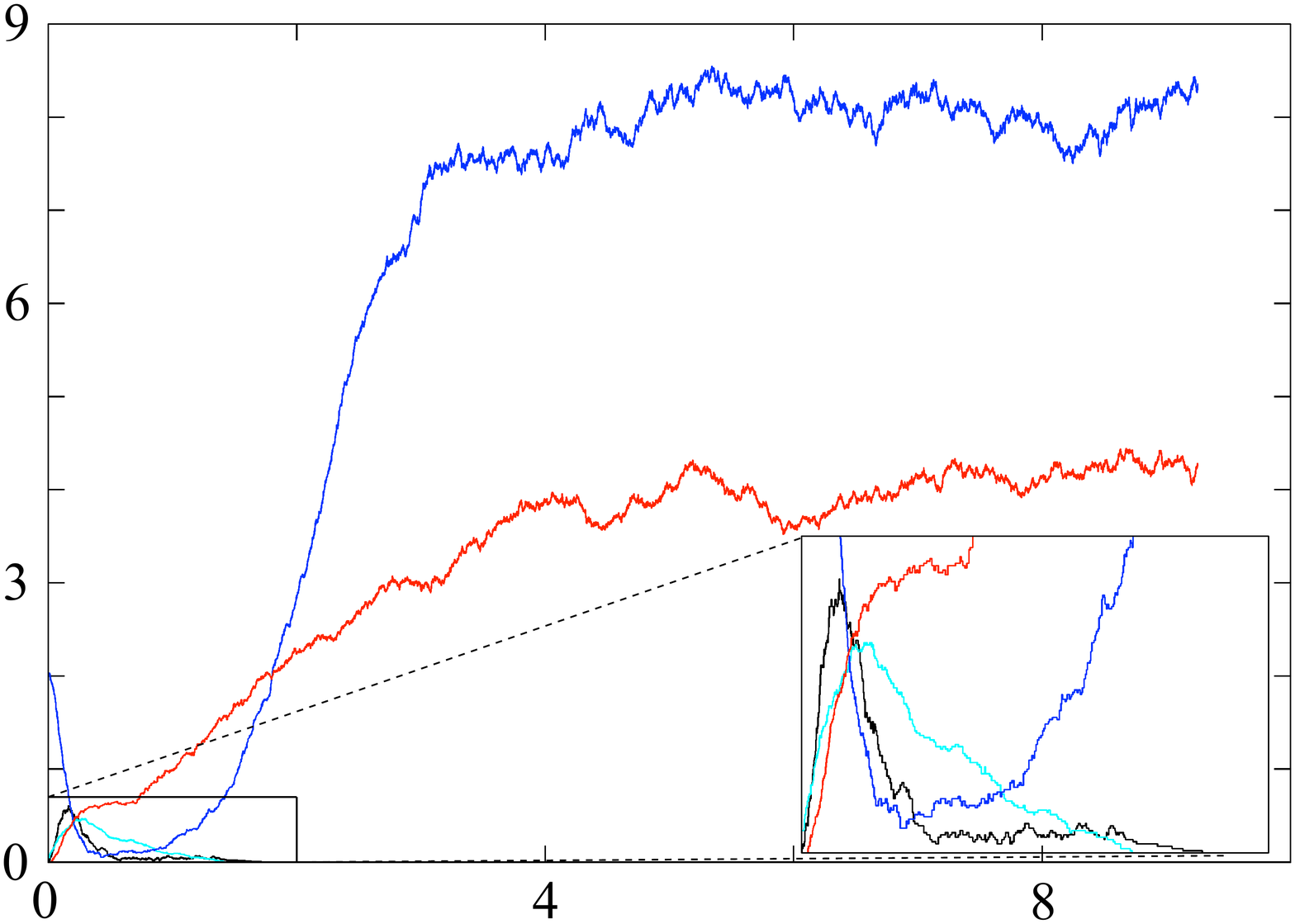}
	\end{minipage}
	\hfill
	\caption{\small 
			Simulations of the evolution of melanoma under T-cell therapy for parameters (\ref{1tcell-muster-parameter}) and initial conditions (\ref{1tcell-muster-incond}).  
			The graphs show the number of individuals divided by 200 versus time. 
			Two scenarios are possible for therapy with T-cells of one specificity: 
			(A) T-cells $z_x$ survive and the system is attracted to $P_{xyz_xw}$, or 
			(B) T-cells $z_x$ die out and the system is attracted to $P_{xy00}$. 
		\label{fig-1tc}}
\end{figure}

\begin{figure}
	\centering
	\begin{minipage}{\textwidth} 
		A \includegraphics[width=0.5\textwidth]{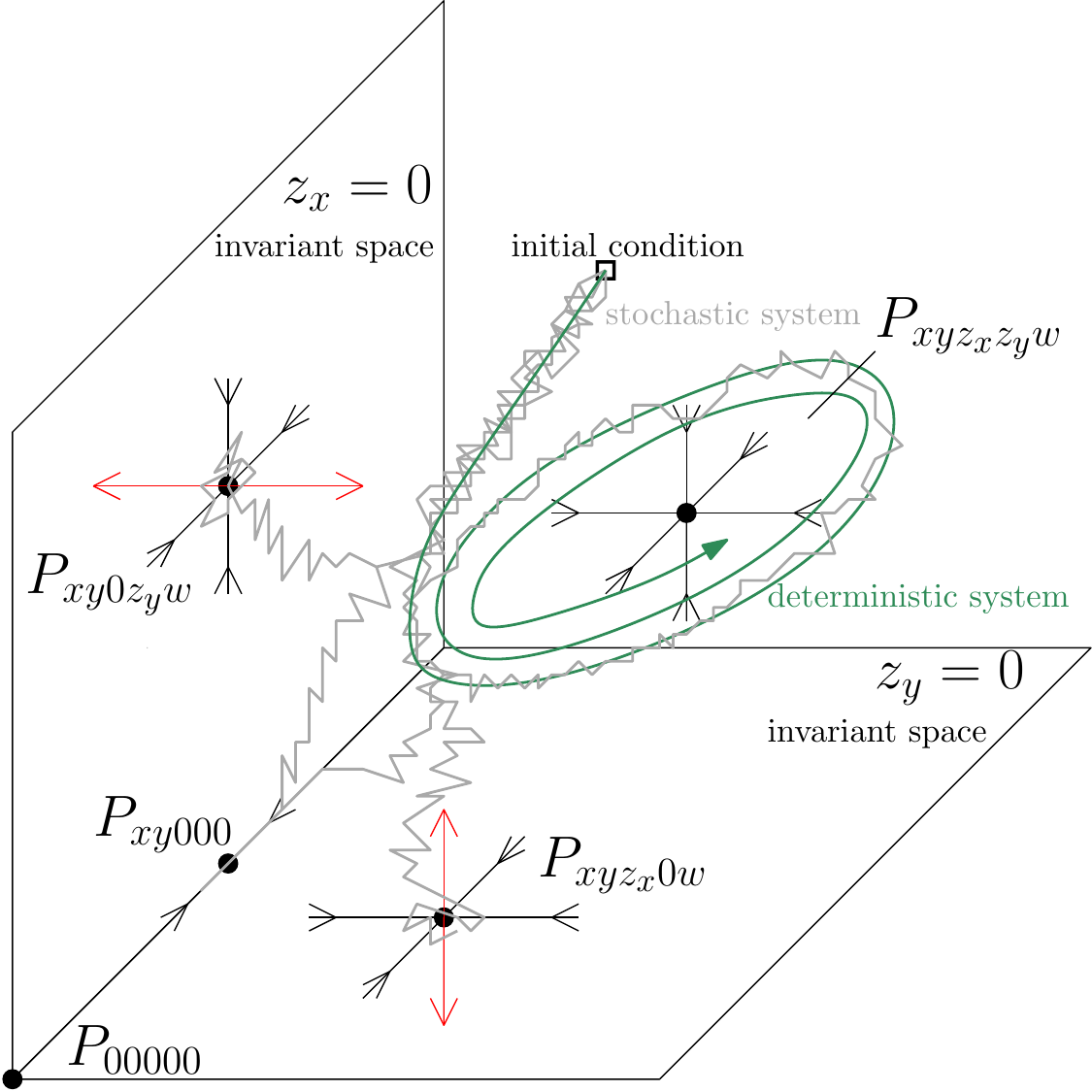}
		\hfill
		B \includegraphics[width=0.4\textwidth]{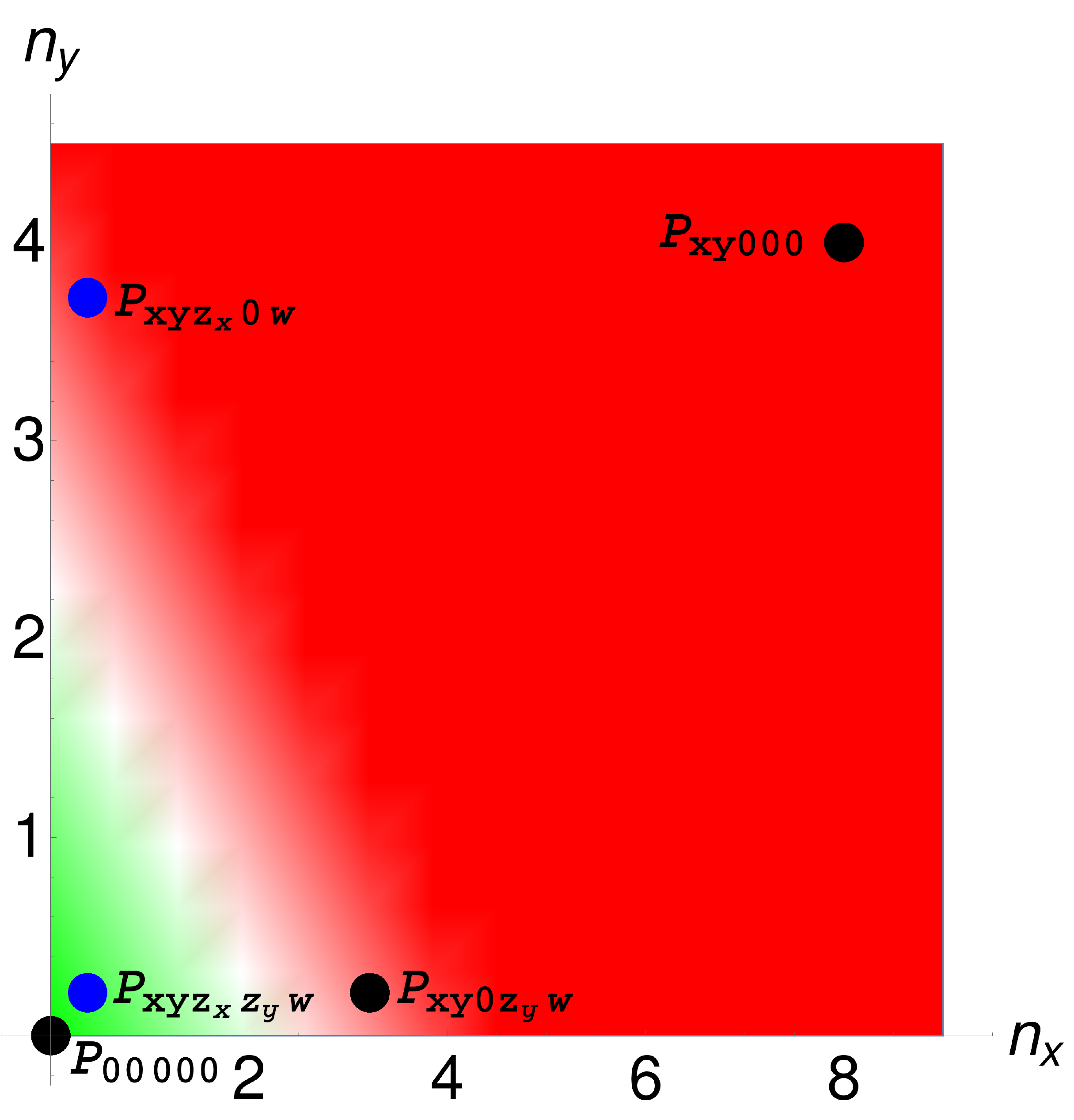}
	\end{minipage}
	\caption{\small {Qualitative description of the role of the fixed points. 
			(A) Sketch of the invariant subspaces, stability of the fixed points, and schematic representation of the dynamics of the deterministic and the stochastic processes. 
			(B) Projection of the fixed points onto $\frakn_x$ and $\frakn_y$. The green area contains the fixed points which correspond to cure or remission and the red area contains those describing relapses.
		}
		\label{fig-fixedpoints}}
\end{figure}

\subsection{Therapy with T-cells of two  specificities}\label{sec-2tcells}

A therapy can only be called successful if the whole tumor is eradicated
or kept small for a long time. A natural idea is thus to inject two types of T-cells in future 
therapies as suggested in \cite{Landsberg:2012vn}. To model this scenario, we just need to add  T-cells attacking the dedifferentiated cells as new actors to the setting described above. We denote them by $z_y$. 
The system contains one more predator-prey term between $y$ and $z_y$:


%
\begin{align}\nonumber
\dot \frakn_{x} &=  \frakn_{x} \Big(b(x)-d(x)-c(x,x) \frakn_{x}\!-c({x
	,  y}) \frakn_y\! -s_w(x,y) \frakn_w\!-s(x,y)-t(z_x,x) \frakn_{z_x}\Big) +s({y,x}) \frakn_y \\\nonumber
\dot \frakn_y &=  \frakn_y \Big(b(y)\!-d(y)\!-c({y,y}) \frakn_y\!-c({y, x}) \frakn_{x}\!-s({y,x}) -t(z_y,y) \frakn_{z_y}\Big)+\big(s_w(x,y) \frakn_w\! +s({x, y})\big) \frakn_{x} \\\nonumber
\dot \frakn_{z_x}\! &=  \frakn_{z_x} (b({z_x,x}) \frakn_{x}-d(z_x))\\\nonumber
\dot \frakn_{z_y} \!&=  \frakn_{z_y} (b({{z_y},y}) \frakn_y-d({z_y}))\\\nonumber
\dot \frakn_w &= -\frakn_w d(w)+
(\ell^{\text{kill}}_w({z_x},x) \,t({z_x},x) 
+\ell^{\text{prod}}_w({z_x},x) \,b({z_x},x) )\frakn_{x}\frakn_{z_x}
\\
&\hspace{2cm}+
(\ell^{\text{kill}}_w({z_y},y) \,t({z_y},y) 
+\ell^{\text{prod}}_w({z_y},y) \,b({z_y},y) )\frakn_{y}\frakn_{z_y}\label{2tcells-system}
\end{align}

In addition to parameters \eqref{1tcell-muster-parameter}, we use the following ones: 
\begin{equation}\label{2tcell-muster-parameter-short}
\begin{array}{lll}
t(z_y,{y})=28 &\ell^{\text{kill}}_w(z_y,y)=1&  d(z_y)=3 				\\
b(z_y,{y})=14 &\ell^{\text{prod}}_w(z_y,y)=0	&
\end{array}
\end{equation}

 and initial conditions:
\begin{equation}\label{2tcell-muster-incond}
\begin{array}{ll}
(\frakn_{x},\frakn_{y},\frakn_{z_x},\frakn_{z_y},\frakn_{w})(0)=(2,0,0,0.05,0.2,0).
\end{array}
\end{equation}

The solution to the deterministic system \eqref{2tcells-system} with parameters \eqref{1tcell-muster-parameter} and \eqref{2tcell-muster-parameter-short} and initial conditions \eqref{2tcell-muster-incond} can be seen on Figure \ref{det-syst} (B). Initial conditions with only T-cells of type $z_y$ or without any T-cells can be seen on Figure  \ref{det-syst} (C) and (D). 

\begin{figure}[h!]
	\centering
	A \includegraphics[width=.47\textwidth]{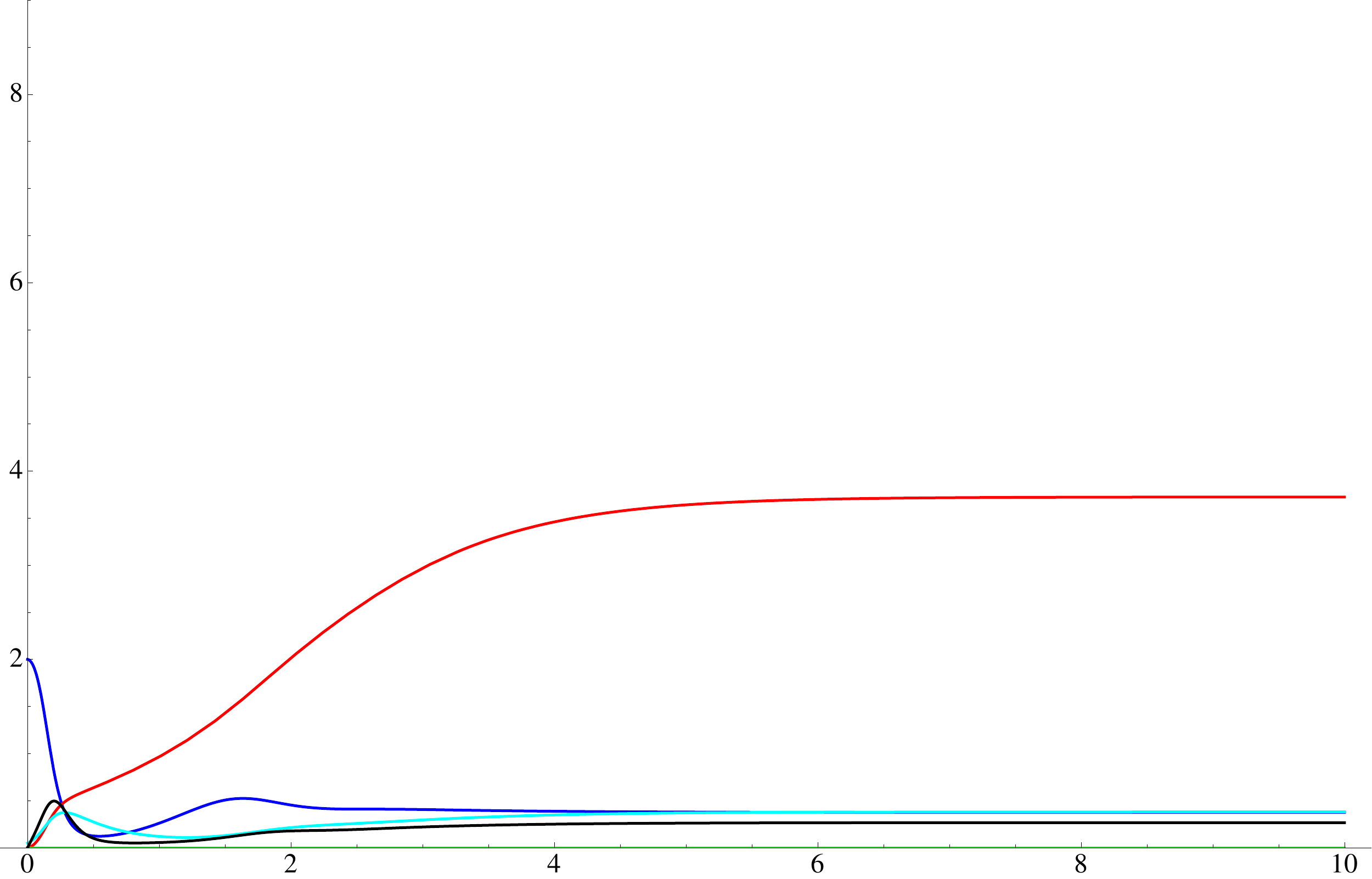}
	B \includegraphics[width=.47\textwidth]{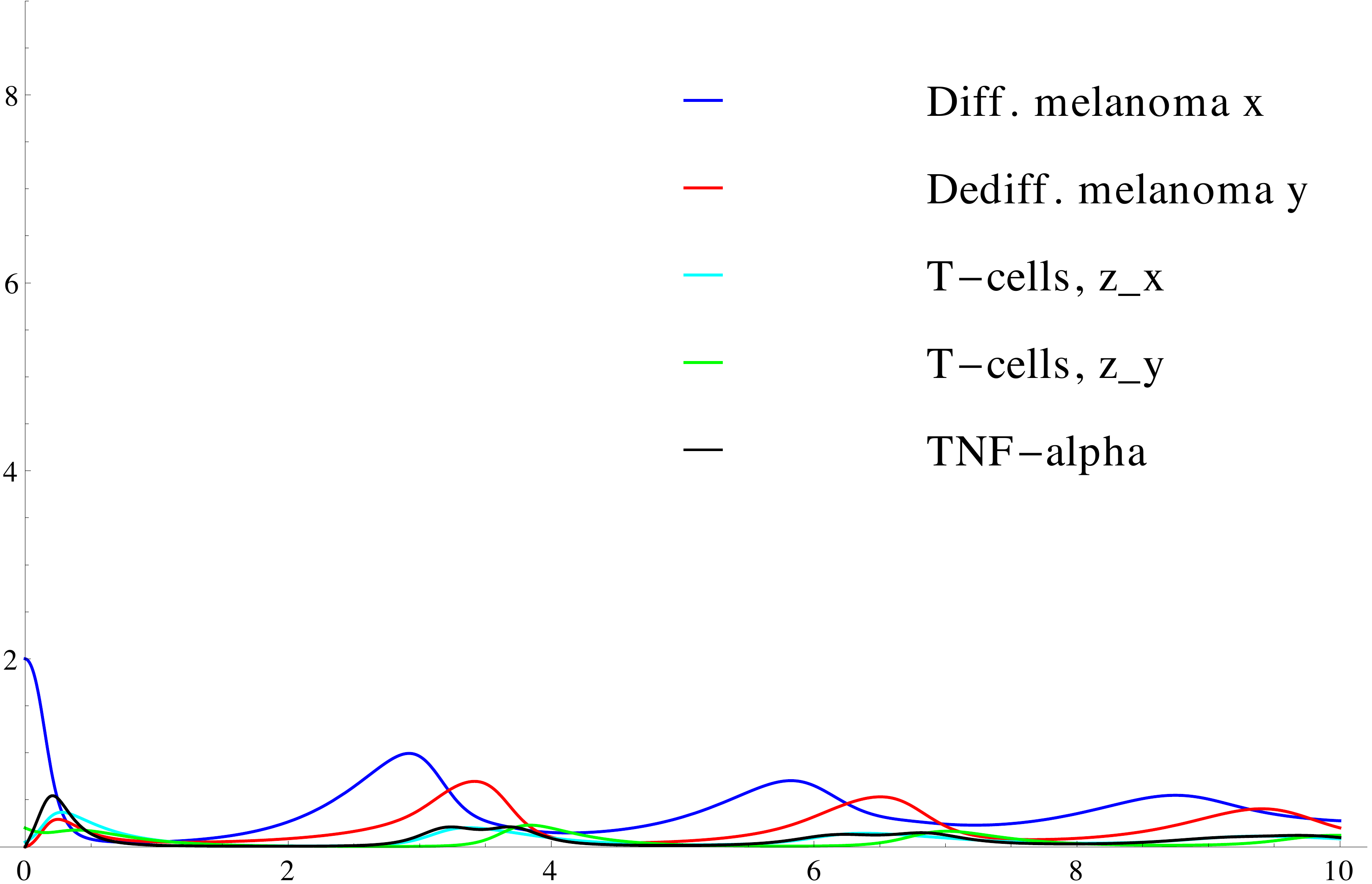}
	C \includegraphics[width=.47\textwidth]{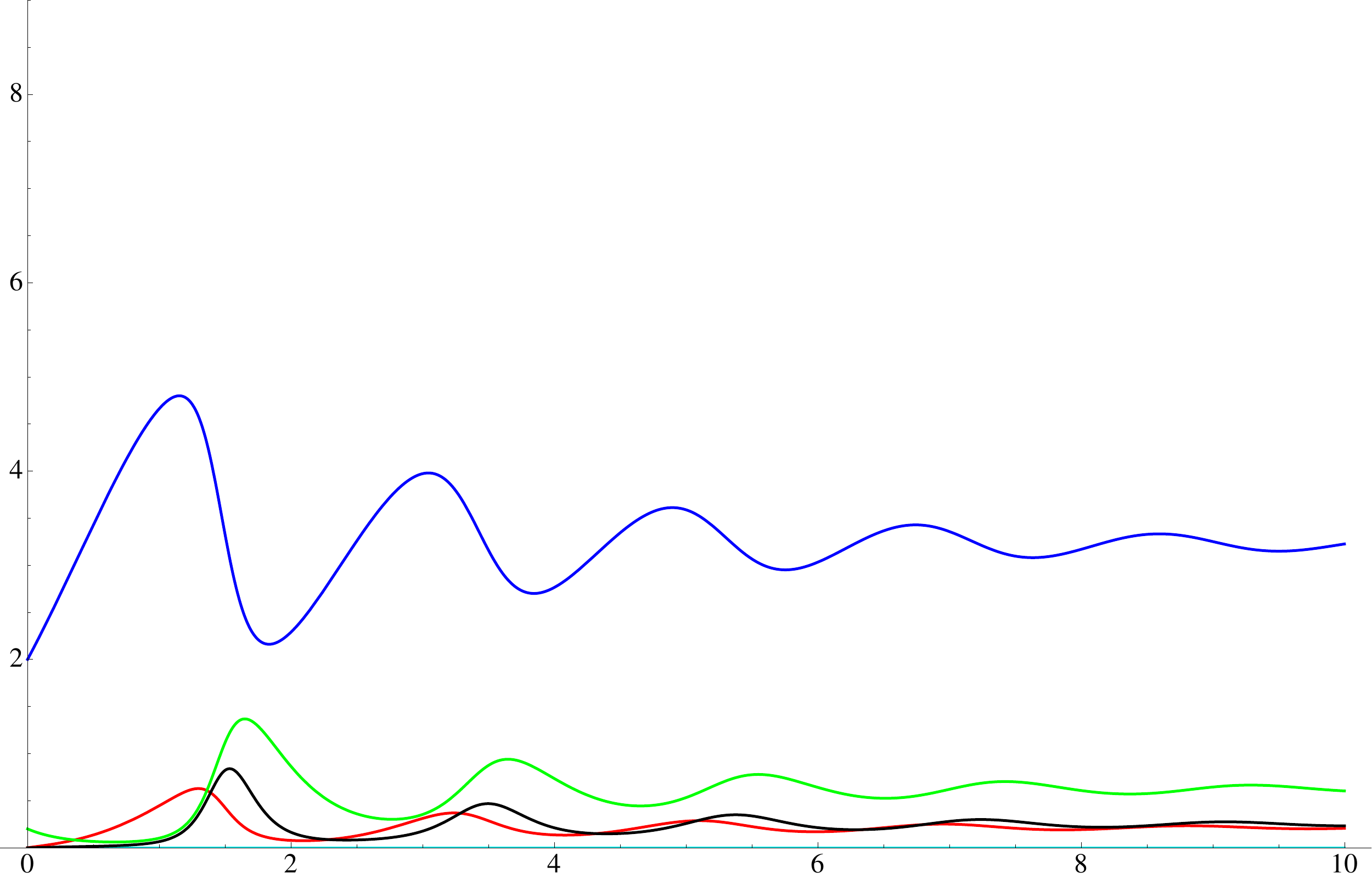}
	D \includegraphics[width=.47\textwidth]{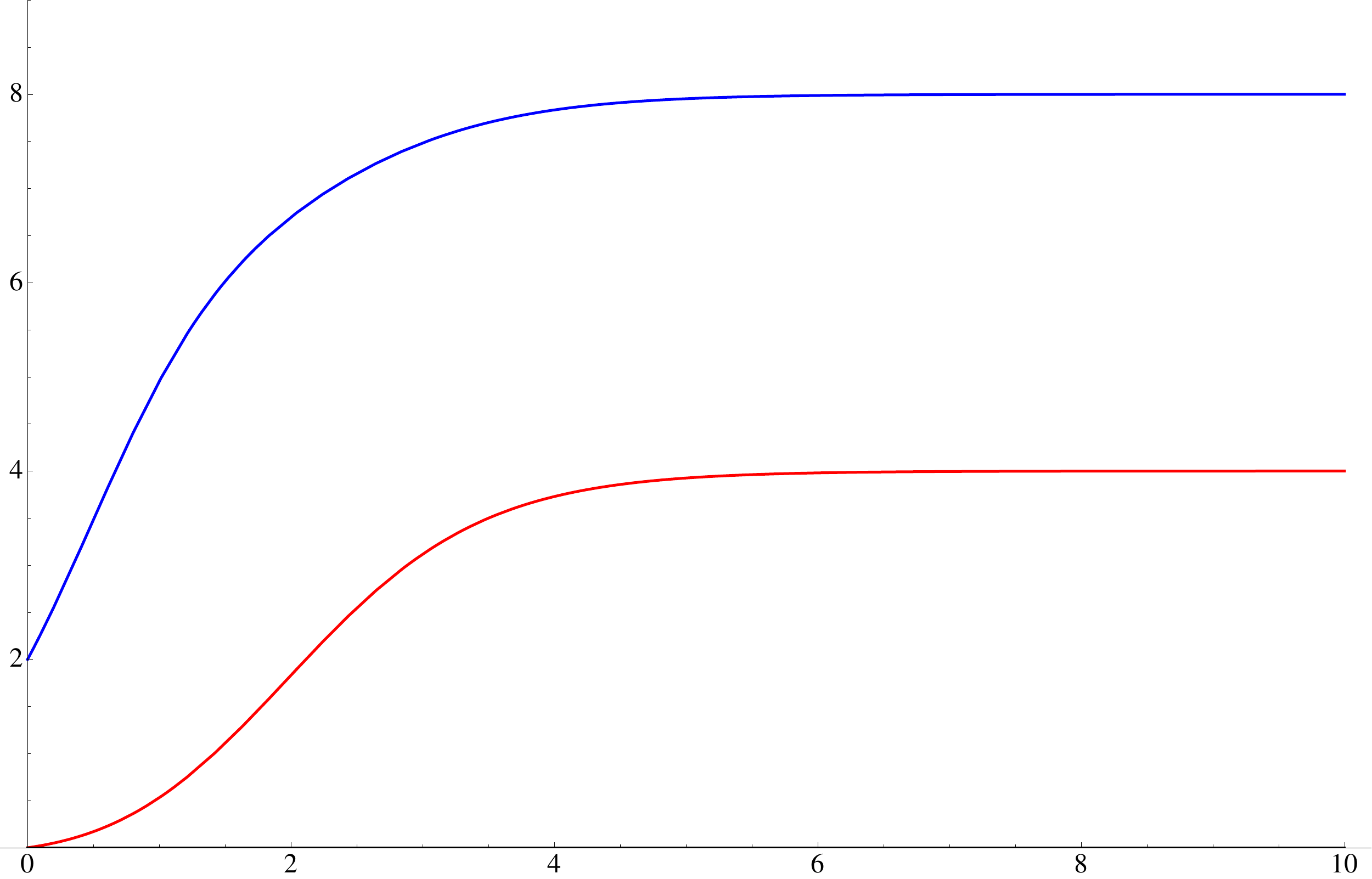}
	\captionof{figure}{\small
		Solutions to the deterministic system \eqref{2tcells-system} with parameters \eqref{1tcell-muster-parameter} and \eqref{2tcell-muster-parameter-short}. With \mbox{$\frakn(0):=(\frakn_{x},\frakn_{y},\frakn_{z_x},\frakn_{z_y},\frakn_{w})(0)$} as the  initial conditions, the system is attracted to the indicated fixed point: \\
		(A) $\frakn(0)=(2,0,0.05,0,0)$, fixed point $P_{xyz_x0w}$,  \\
		(B) $\frakn(0)=(2,0,0.05,0.2,0)$,  fixed point $P_{xyz_xz_yw}$,\\
		(C) $\frakn(0)=(2,0,0,0.2,0)$,  fixed point $P_{xy0z_yw}$,\\
		(D) $\frakn(0)=(2,0,0,0,0)$,  fixed point $P_{xy000}$.
	} \label{det-syst} 
\end{figure}

\begin{figure}[h!]
	\centering
	\begin{minipage}{.49\textwidth} 
		A \includegraphics[width=\textwidth]{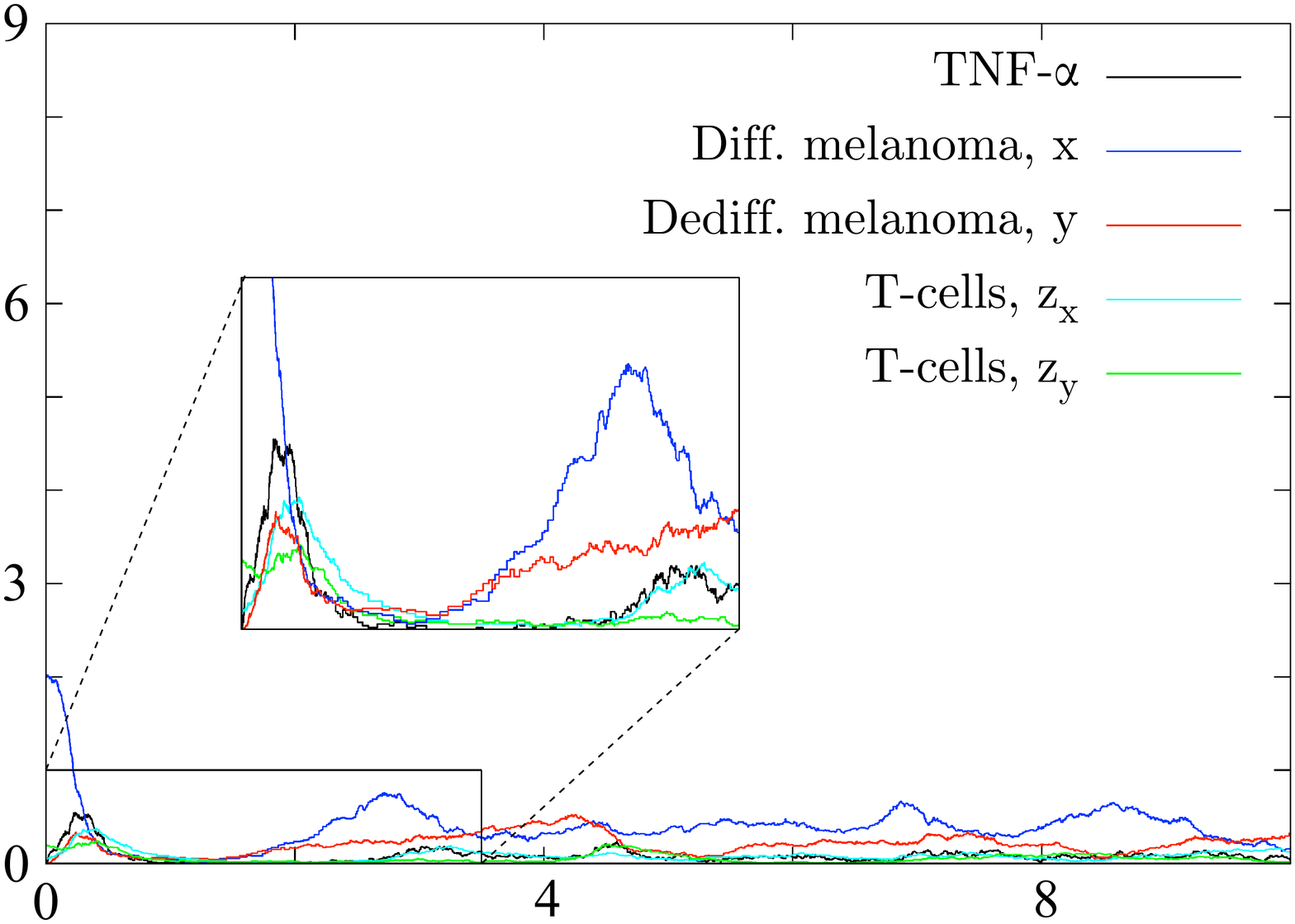}
	\end{minipage}
	\hfill
	\begin{minipage}{.49\textwidth} 
		B \includegraphics[width=\textwidth]{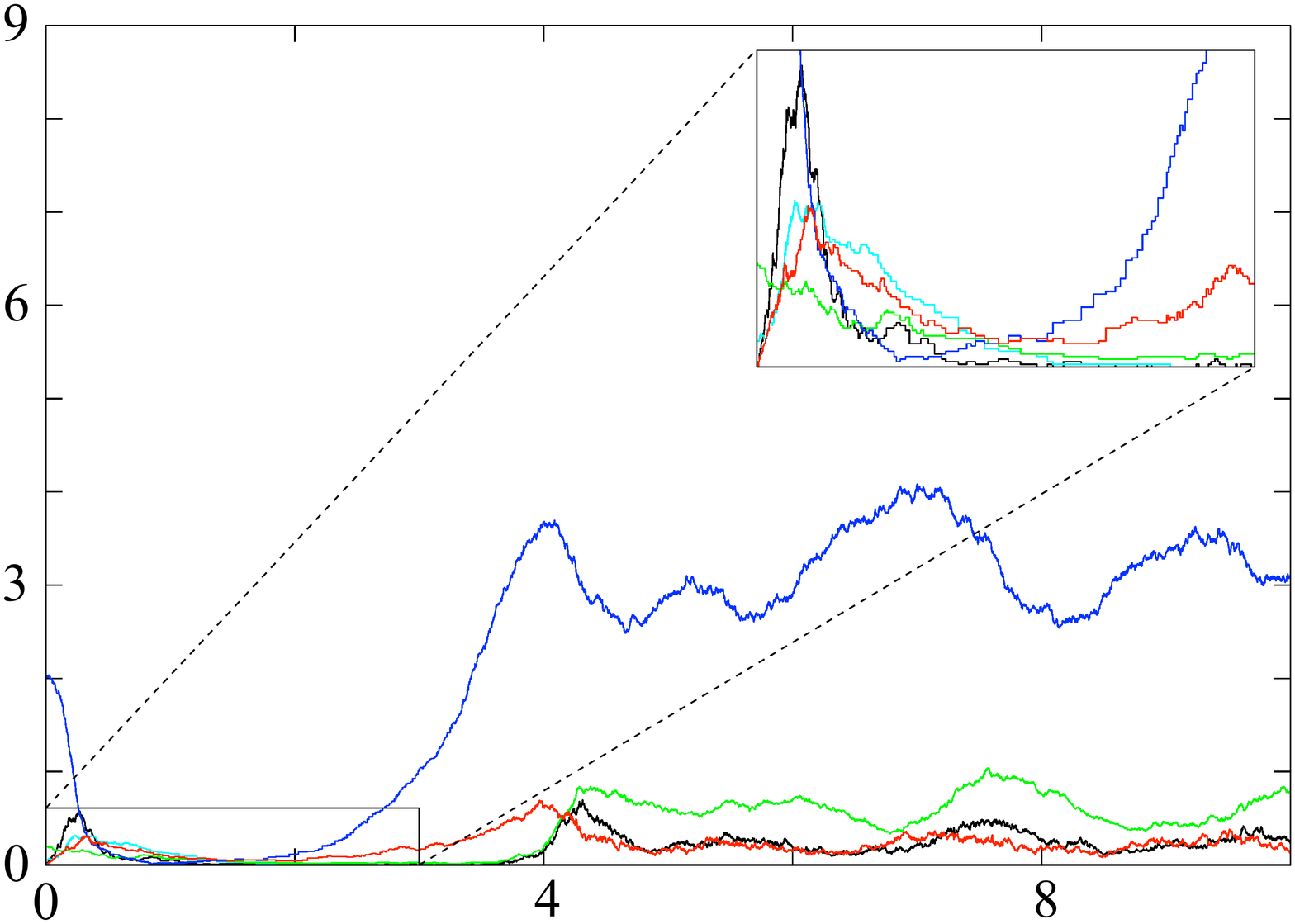}
	\end{minipage}
	\hfill
	\begin{minipage}{.49\textwidth} 
		C \includegraphics[width=\textwidth]{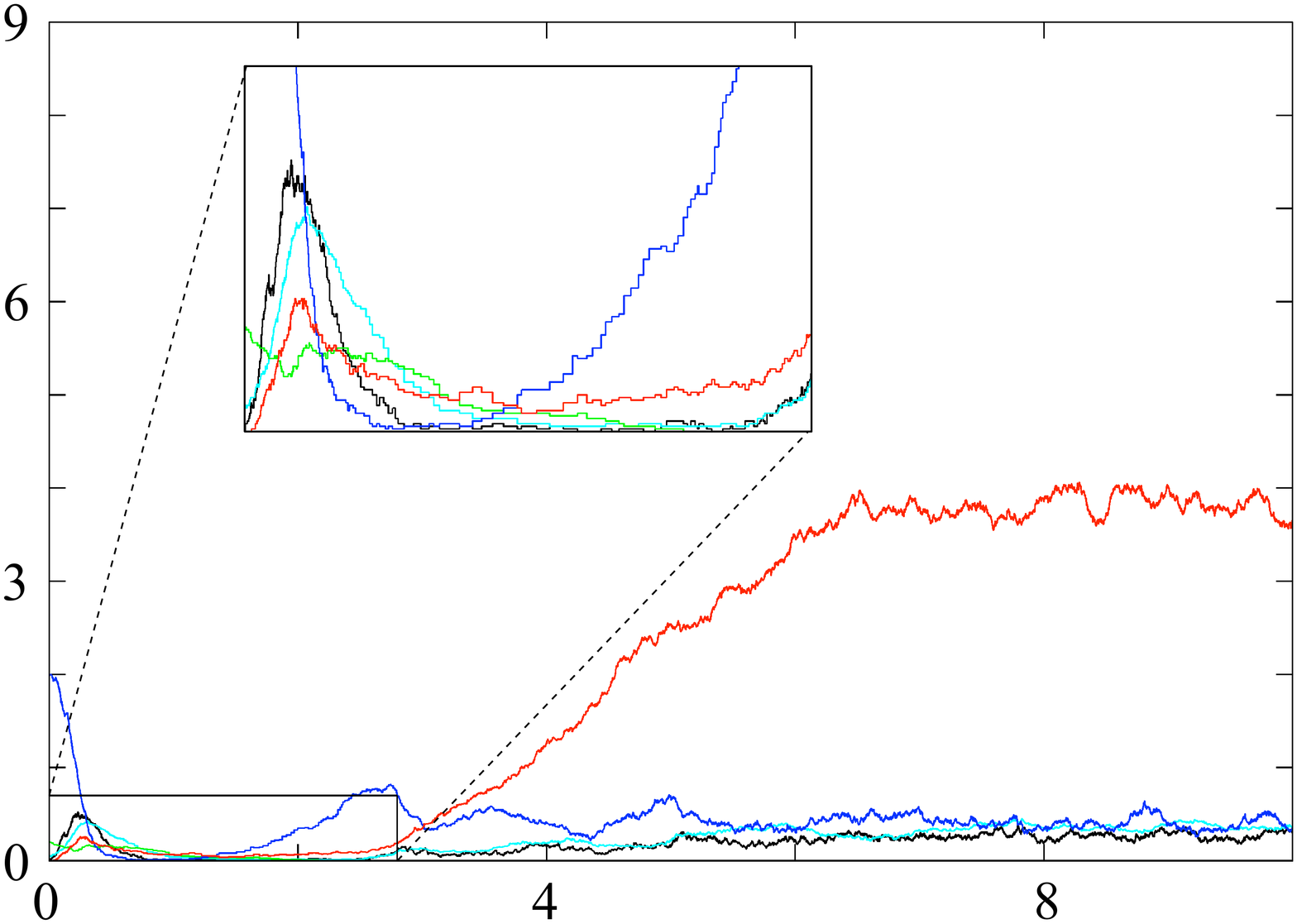}
	\end{minipage}
	\hfill
	\begin{minipage}{.49\textwidth} 
		D \includegraphics[width=\textwidth]{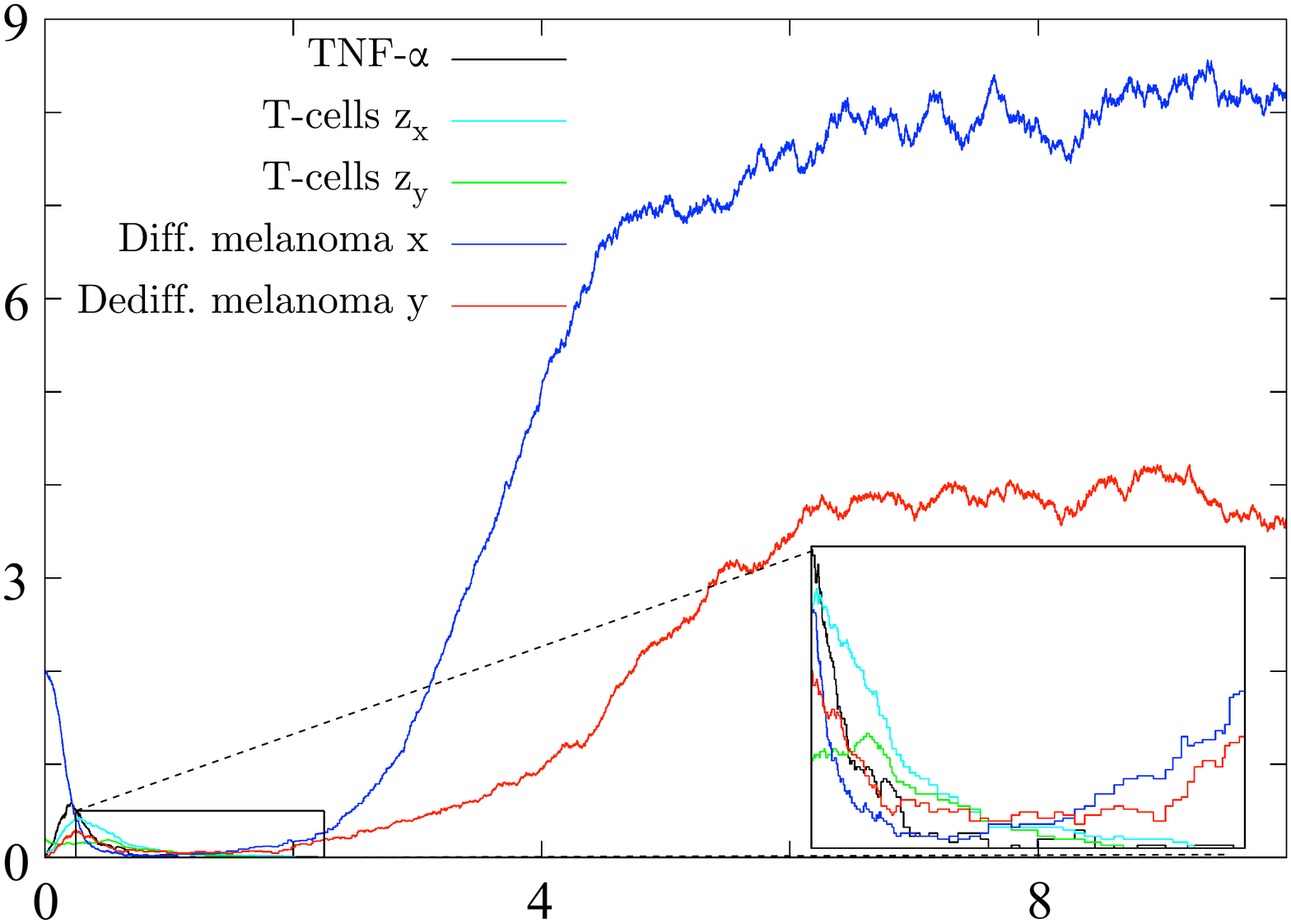}
	\end{minipage}
	\hfill
	\begin{minipage}{.49\textwidth} 
		E \includegraphics[width=\textwidth]{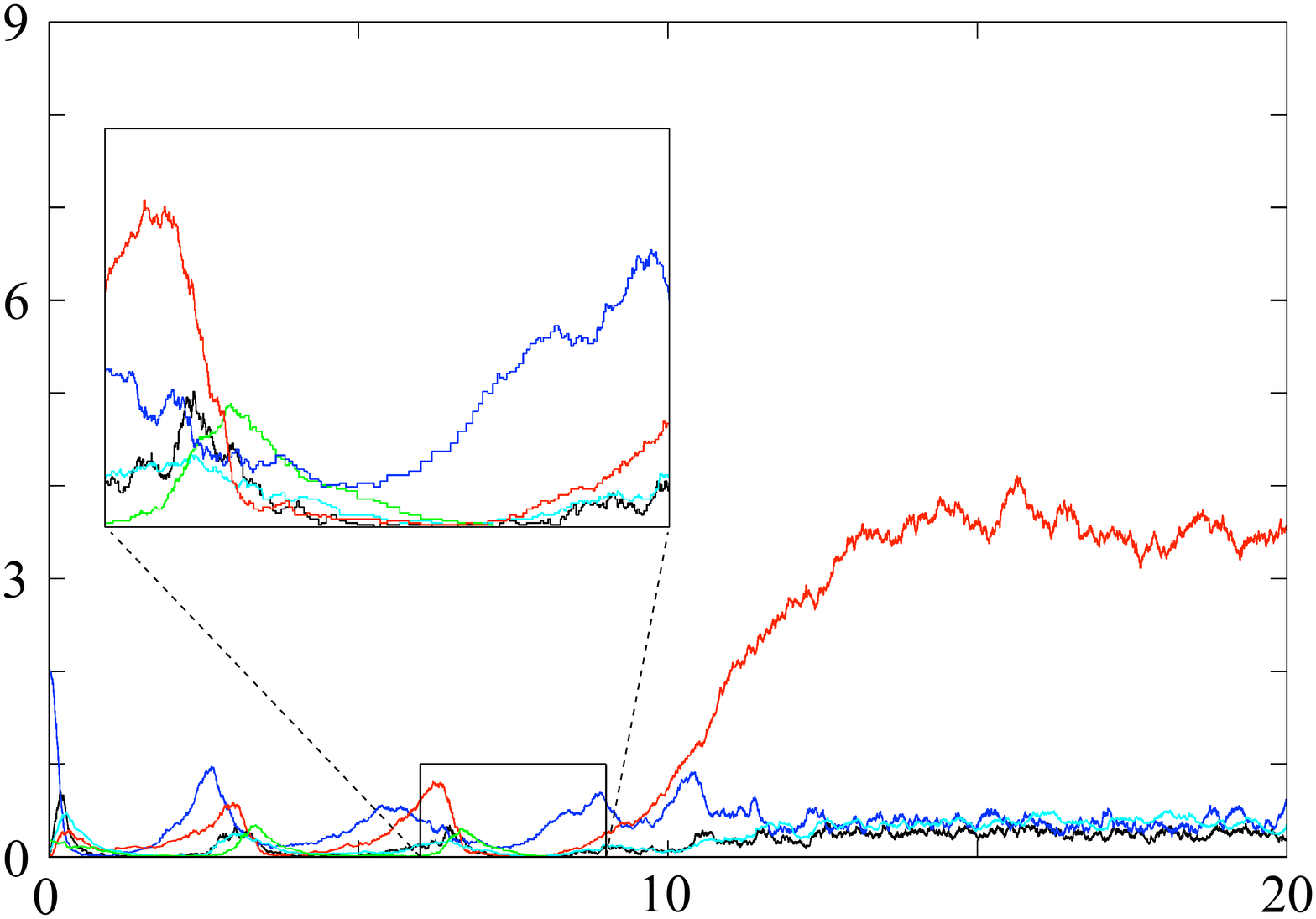}
	\end{minipage}
	\hfill
	\begin{minipage}{.49\textwidth} 
		F \includegraphics[width=\textwidth]{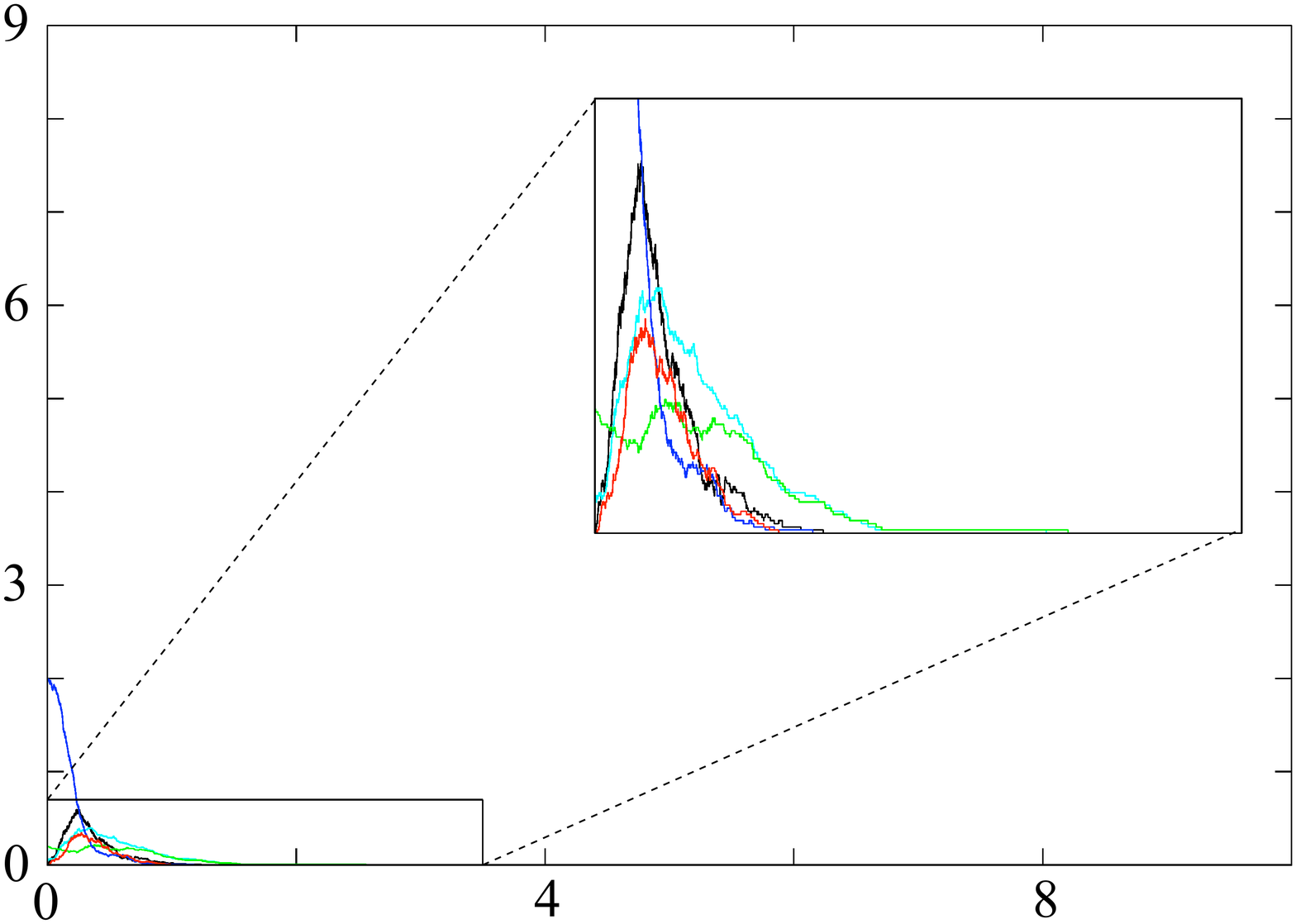}
	\end{minipage}
	\hfill
	\caption{\small  
		Simulations of the stochastic evolution of melanoma under T-cell therapy for parameters \eqref{1tcell-muster-incond} and \eqref{2tcell-muster-parameter-short}.  The graphs show the number of individuals divided by 200 versus time. 
		Possible scenarios for therapy with T-cells of two specificities:
		(A) T-cells $z_x$ and $z_y$ survive and the system stays close to $P_{xyz_xz_yw}$,
		(B) T-cells $z_x$ die out and the system is attracted to $P_{xy0z_yw}$,
		(C) T-cells $z_y$ die out and the system is attracted to $P_{xyz_x0w}$, 
		(D) Both T-cell types $z_x$ and $z_y$ die out and the system is attracted to $P_{xy000}$.
		(E) Transition between cases (A) and (C).
		(F) the tumor is eradicated (corresponding to $P_{00000}$).
		\label{stoch-bif}}
\end{figure}

The introduction of $z_y$ adds two new fixed points (see the blue dots on Figure  \ref{fig-fixedpoints} (B)): $P_{xyz_xz_yw}$ is the new stable fixed point with all non-zero populations, and $P_{xy0z_yw}$ corresponds to the absence of the T-cell population of type $z_x$. The invariant subspaces are now $\{\frakn_{z_x}=0\}$, in which $P_{xy0z_yw}$ is stable, $\{\frakn_{z_y}=0\}$, in which $P_{xyz_x0w}$ is stable and $\{\frakn_{z_x}=0\}\cap\{\frakn_{z_y}=0\}$, in which $P_{xy000}$ is stable. Note that  $P_{xyz_x0w}$, corresponding to $P_{xyz_xw}$ from the last section, is unstable in the enlarged space.

With the same initial conditions as before and $\frakn_{z_y}(0)$ small but positive, the \emph{deterministic} system is  attracted to the stable fixed point $P_{xyz_xz_yw}$: 
the T-cell population, $\frakn_{z_x}$, increases in presence of its target $x$, 
TNF-$\alpha$ is secreted, and the
differentiated melanoma population shrinks due to killing and switching, the population of dedifferentiated melanoma grows, but is regulated and kept at a low level by the T-cells of type $z_y$. Similarly, $\frakn_x$ is regulated by $\frakn_{z_x}$. 

We choose the parameters such that the minima of the two types of
T-cells during remission are low, so that they have a large enough probability to die out in the stochastic system. 
Since at the beginning of therapy no or only very few dedifferentiated melanoma cells are present, the population of T-cells of type $z_y$ starts growing only later. In order to avoid their early extinction a higher initial amount of these T-cells can be injected.
There are now five main different scenarios in the stochastic system (see Figure  \ref{stoch-bif}).
Either the T-cells of type $z_x$  (B), or the T-cells of type $z_y$ (C), or both of them die out (D). 
Also all populations can survive for some time fluctuating 
around their joint equilibrium (A). The fifth scenario is a cure, i.e.\ the extinction of the entire tumor due to the simultaneous attack of the two different T-cell types (F). T-cells and TNF-$\alpha$ vanish since 
they are not produced any more in the absence of their target. Of course, transitions between the different scenarios are also possible, e.g. the system could pass from Case (A) to (B) or (C) and then to (D), see Figure  \ref{stoch-bif} (E). Furthermore, note that setting the switch from $x$ to $y$ to zero introduces an additional scenario: it is then possible that a relapse appears, which consists only of differentiated melanoma cells.

Starting from our choice of initial conditions, the deterministic system converges to $P_{xyz_xz_y}$, but the stochastic system can hit one of the invariant hyperplanes due to fluctuations, and is driven to different possible fixed points, see Figure  \ref{fig-fixedpoints} (A). 
The transitions between the different scenarios can be seen as a metastability phenomenon.

\subsection{Reproduction of experimental observations and predictions}\label{sec-bio-param}
\subsubsection{Comparison of experimental observations and simulations}
The parameters of Section \ref{sec-relapse}
are chosen ad hoc to highlight the influence of randomness and the possible behavior of the system. 
Let us now show that our models are capable to reproduce the experimental data of Landsberg et al. \cite{Landsberg:2012vn} \emph{quantitatively}. The  choice of parameters is explained below (Subsection~\ref{subsec-parameters}).

Figure  \ref{Data} (A) shows the experimental data of \cite{Landsberg:2012vn} whereas Figure  \ref{Data} (B) shows  the results of our simulations. Each curve describes the evolution of the diameter of the tumor over time. In the 
stochastic system two situations can occur: first, the relapse consists 
mainly of differentiated melanoma cells and the tumor reaches its original size again after $90$ days. 
This is the case if the T-cells die out. Second, the 
relapse 
consists mainly of dedifferentiated cells and  the tumor reaches its original size again after roughly $190$ 
days. This is the case if the T-cells survive the phase 
of remission, become active again and kill differentiated cancer cells. In the simulations  the therapy with one type of T-cells pushes the tumor down to a microscopic level for $50$ 
to $60$ 
days, as in the experimental data. 
The curves marked ACT in the experimental data in Figure \ref{Data} (A) 
are matched by simulation data when the T-cells die out (Differentiated relapse in Figure \ref{Data} (B)). 
In the experiments there might be T-cells which lose their function, e.g.\ due to exhaustion, and cannot kill the differentiated melanoma cells. This effect is to be seen as included in the death rate of T-cells in the model. They can be re-stimulated and become active again which is marked as ACT+Re  in Figure \ref{Data} (A). Although our model does not include re-stimulation, the case of surviving T-cells in the simulations (Dediff.\ relapse in Figure \ref{Data} (B)) can qualitatively be interpreted as the case of ACT+Re.  Note that the scales of the axes are the same in both figures and that the experimental findings are met very well by the simulations. The simulated curves under treatment start at the beginning of the treatment and not at day zero. The detailed pictures showing the evolution of melanoma and T-cell populations during the therapy are given in Figure \ref{Biopictures}.

\begin{figure}[h!]
	\centering
	\begin{minipage}{.49\textwidth} 
		\vspace{-0.4cm}
		A \includegraphics[width=\textwidth]{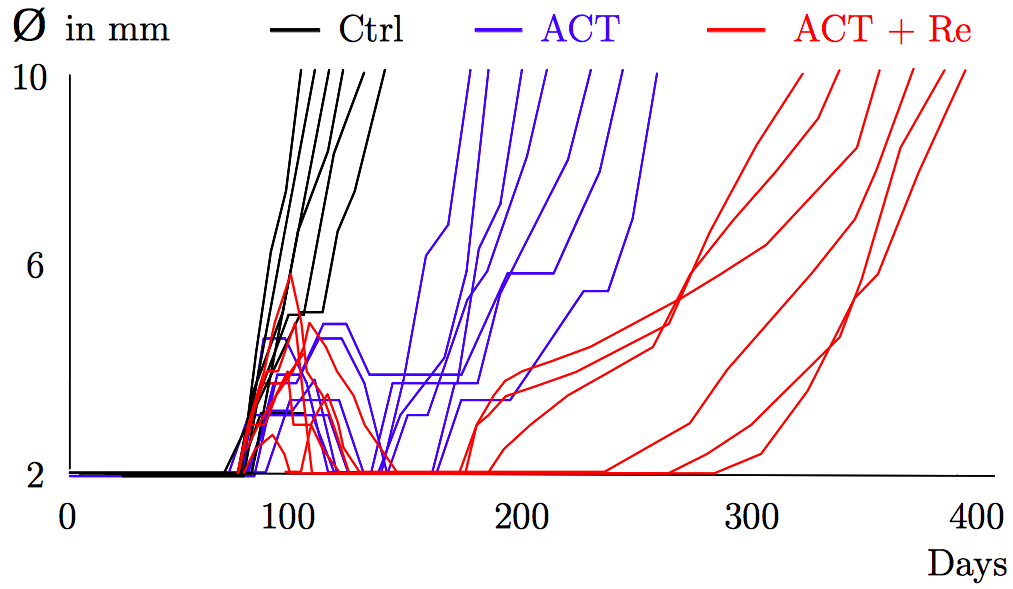}
	\end{minipage}\hfill
	\begin{minipage}{.49\textwidth} 
		B \includegraphics[width=\textwidth]{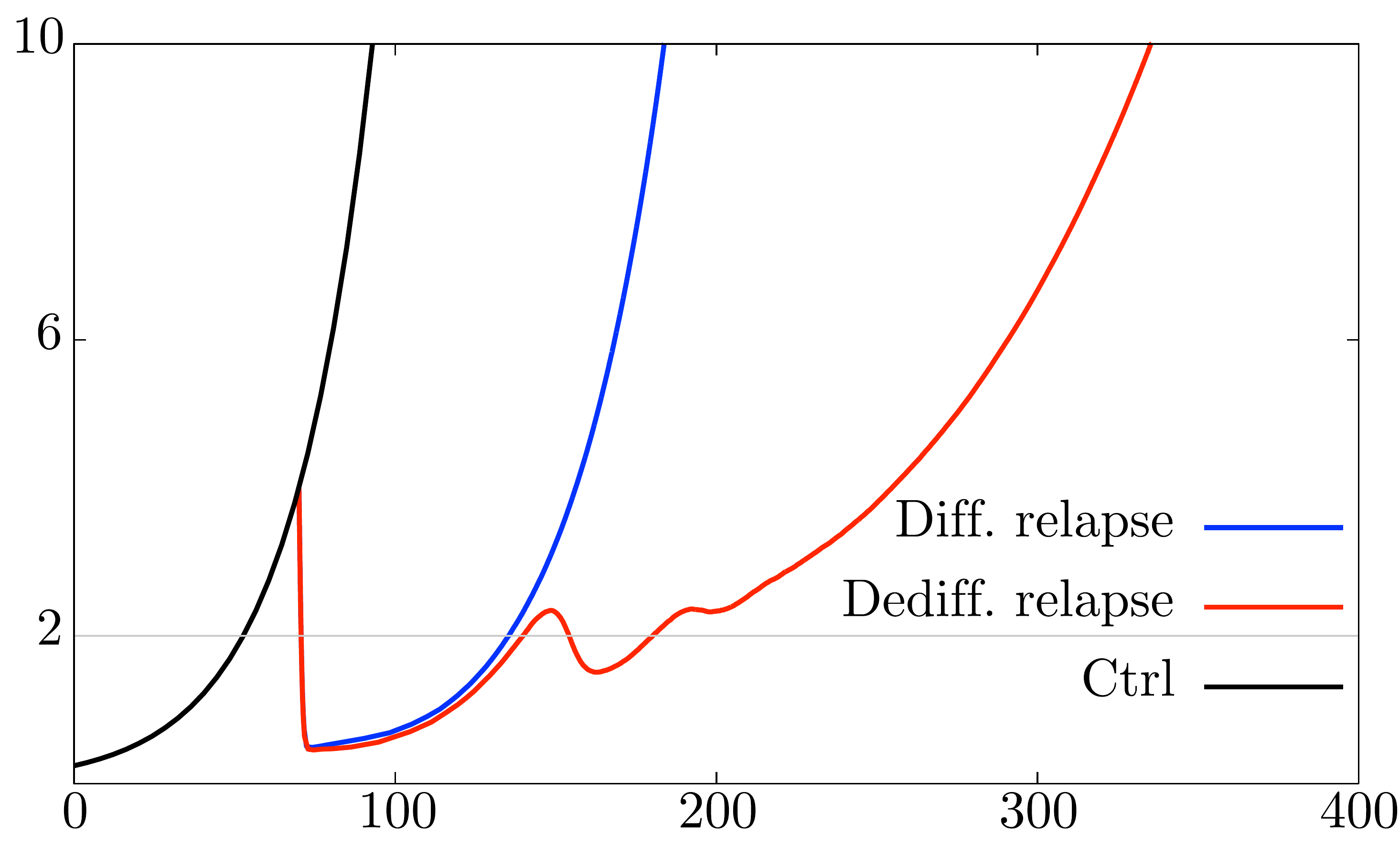}
	\end{minipage}
	\caption{
		\small {
			Comparison of experimental data obtained by Landsberg et al.\ with simulations for biologically reasonable parameters.
			The graphs show the diameter of the tumor measured in millimeters versus time in days after tumor initiation: (A) experimental data, (B) simulated data ($K=10^5$ and $\frakn_{z_x}(0)=0.02$).\label{Data}}}
\end{figure}

\begin{figure}[h!]
	\centering
	A \includegraphics[width=.46\textwidth]{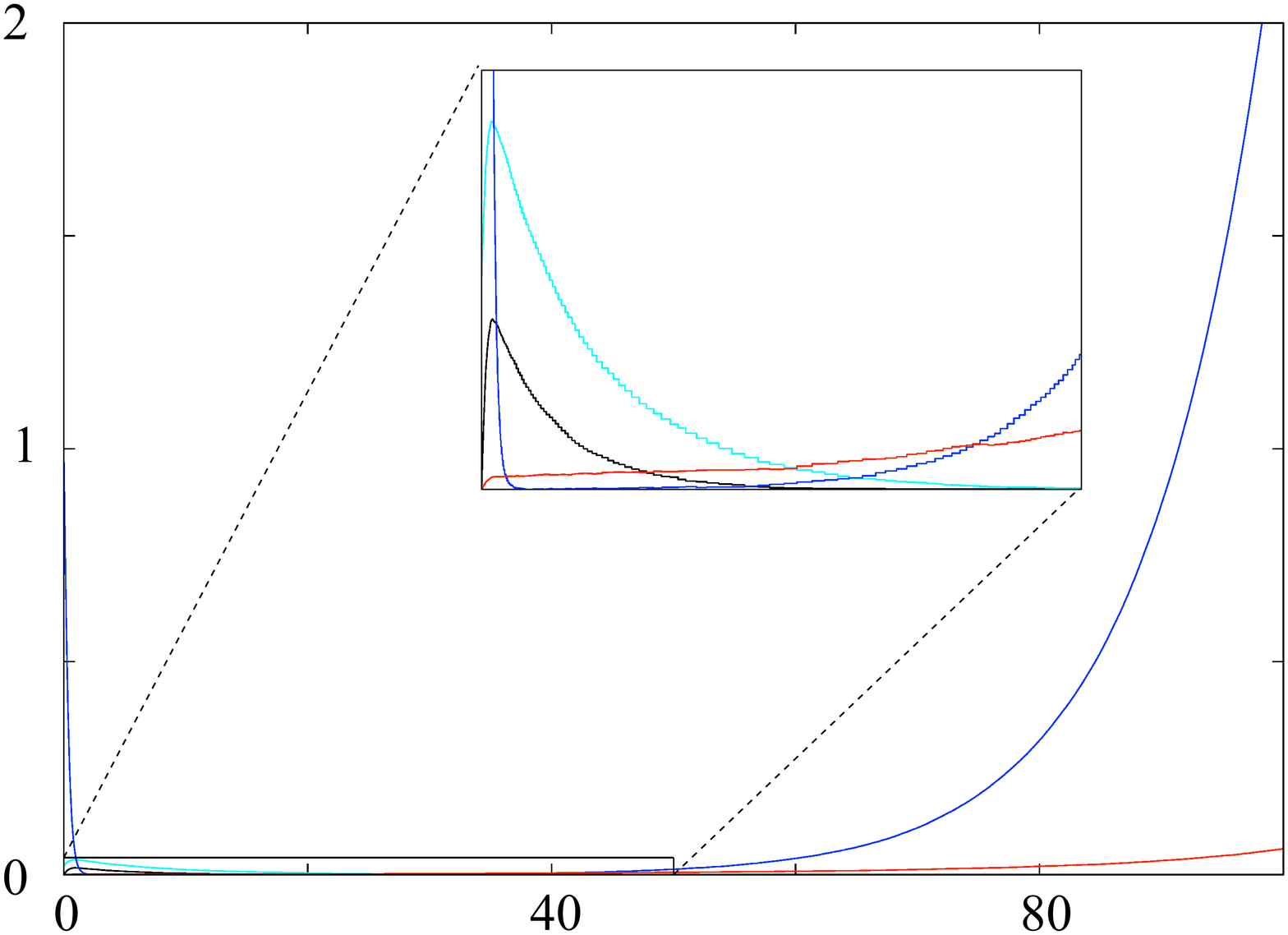}\hfill
	B \includegraphics[width=0.46\textwidth]{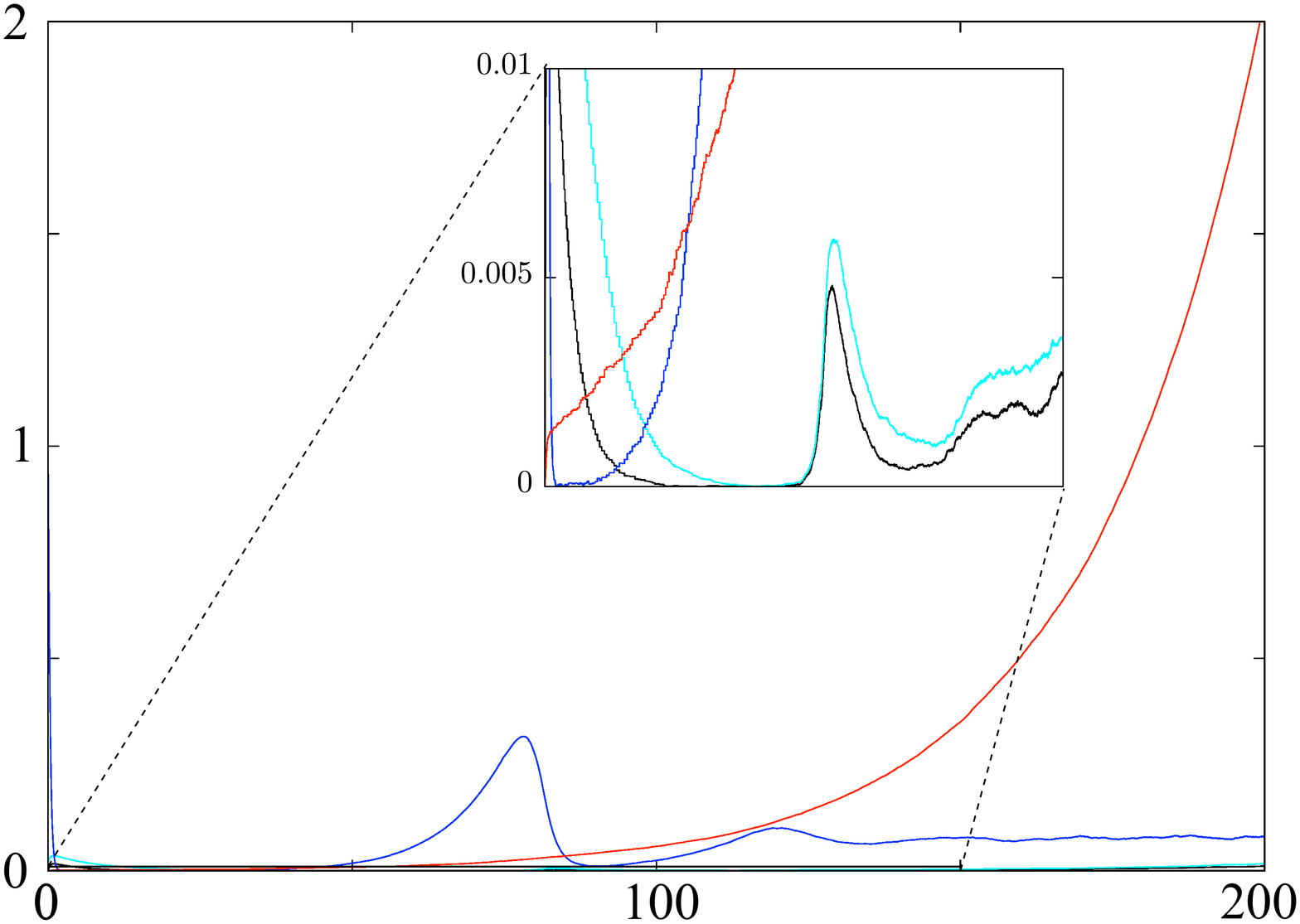}\\
	C \includegraphics[width=.46\textwidth]{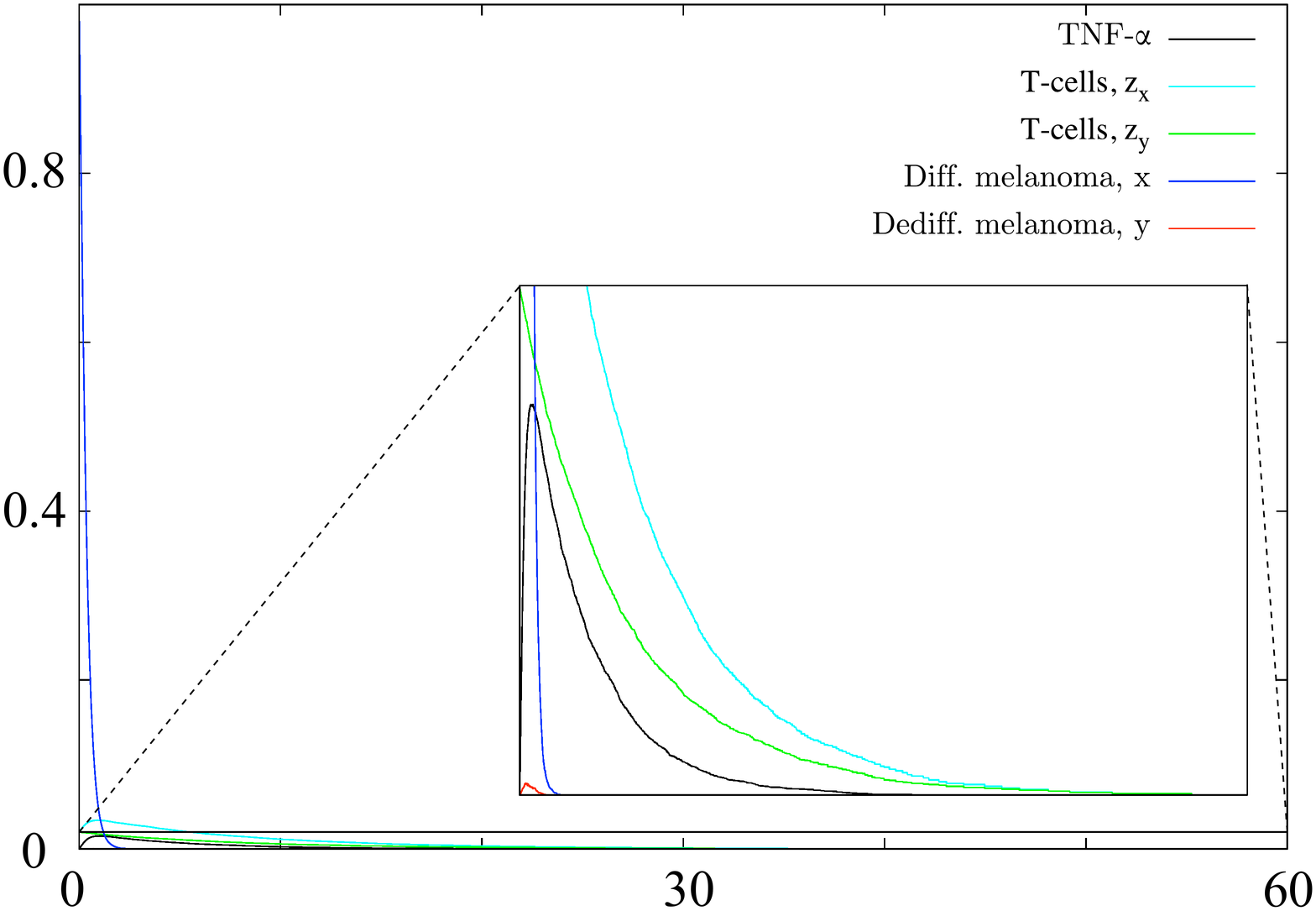}\hfill
	D \includegraphics[width=0.46\textwidth]{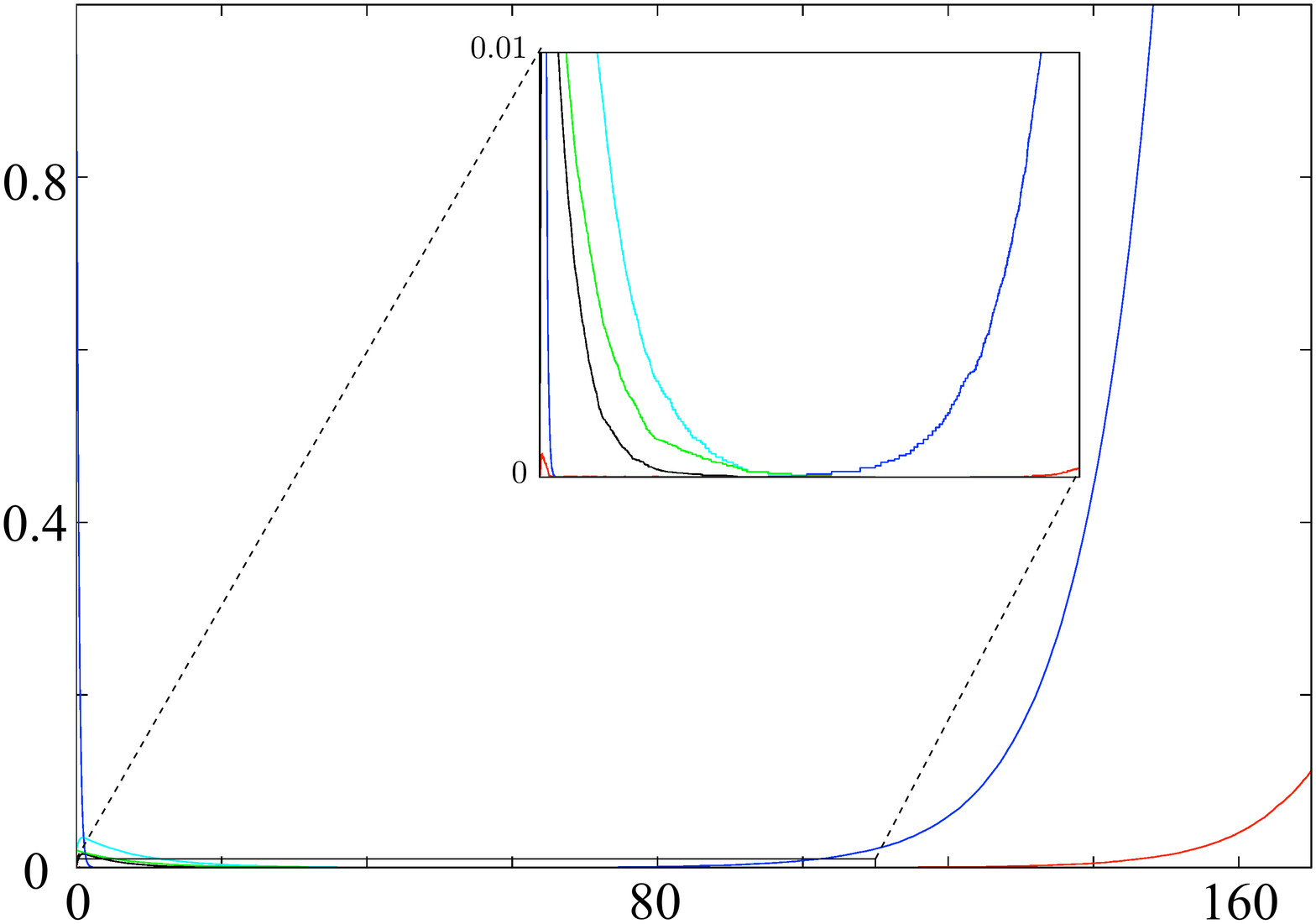}
	\captionof{figure}{\small
		Simulations for biological parameters.
		Therapy with one T-cell type:
		(A) differentiated relapse (T-cells $z_x$ die out),
		(B) dedifferentiated relapse (T-cells $z_x$ survive),
		Therapy with two T-cell types:
		(C) cure,
		(D) differentiated relapse (both T-cell types die out).
	} \label{Biopictures} 
\end{figure}

As there is no data for the case of two T-cells, numerical simulations of such a therapy strategy should be seen as 
predictions. For the new T-cell population (of type ${z_y}$) we choose the same parameters as for the first population (of type ${z_x}$), just the target is 
different. The therapy seems to be very promising: almost all simulations show a cure for these 
parameters, only very few times a relapse occurs. Nevertheless,  
the behavior of the system (e.g.\ the probability to end up in the different scenarios) depends strongly on the 
choice of certain parameters, as pointed out in the last two sections. In order to give a  reliable prediction we 
need data to obtain safer estimates for the most important parameters, which seem to be the switching 
and therapy rates as well as initial values. 

The initial values play an important role for the success of a therapy. In the case of 
therapy with T-cells of one specificity, increasing the initial amount of T-cells has the following effect: 
the melanoma cells are killed faster, the population of differentiated melanoma cells 
reaches a lower minimum and as a consequence the T-cells pass through a lower and broader 
minimum. The probability that the T-cells die out increases,
and a differentiated relapse is  more likely than in the case
of a smaller initial T-cell population. Moreover, the broadening of the minima causes a ``delay'' and both kind of 
relapses (consisting mainly of differentiated or dedifferentiated cells) appear later. 
But since the extinction of T-cells is more likely,  the tumor may reach its original size earlier, see Figure  \ref{fig-initdose}.
For an initial value ten times as large as in Figure \ref{Data} (B) the probability of an eradication of the tumor is still very small.
If the number of T-cells initially is half the number of tumor cells, the probability of a favorable outcome is much higher.
But 
such a high amount of T-cells is unrealistic.

\begin{figure}[h!]
	\centering
	\begin{minipage}{.49\textwidth} 
		C \includegraphics[width=\textwidth]{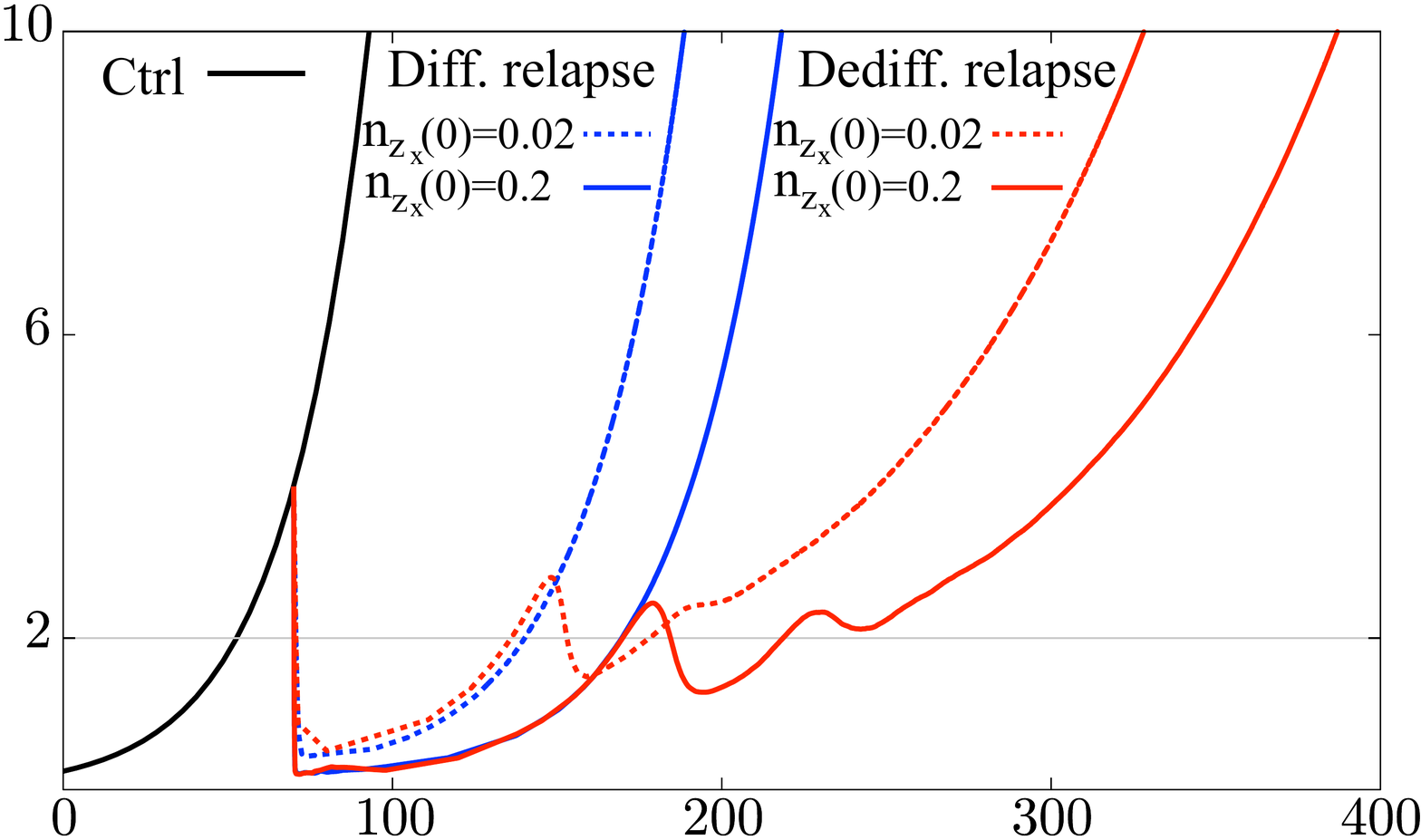} 
	\end{minipage}
	\caption{
		\small {
			Simulations for different initial doses of T-cells: $\frakn_{z_x}(0)=0.2$ and $\frakn_{z_x}(0)=0.02$.\label{fig-initdose}}}
\end{figure}

\subsubsection{Physiologically reasonable parameters}\label{subsec-parameters}

We explain here how we choose the biological parameters. Some parameters can be estimated from the
experimental data. Recall that the subject of \cite{Landsberg:2012vn}
is to investigate the behavior of  melanoma under T-cell therapy in
mice. Without therapy the tumor undergoes only natural birth, death
and switch events.

\begin{itemize}
	\item \textit{Choice of birth and death rates:} We assume that the number of cells in the tumor is described by   
	\begin{equation}\label{estimation-growth}
	N_t\approx N_0\exp (rt),
	\end{equation}
	where $N_t$ denotes the number of cells at time $t$, $N_0$ the initial population size and $r$ the overall 
	growth rate. Note that the estimate of the growth rate is independent  of the initial value. Figure 4 (A) in the main part shows that the tumor needs roughly 50 days
	(without therapy) to grow from 2 mm diameter to 
	10 mm diameter. Since the structure of a melanoma is 3-dimensional this corresponds roughly to $N_{50}=125 
	N_0$ which implies $r=0.1$. Unfortunately, no data that allow to estimate the ratio of birth and death 
	events are provided. As long as mutations are not considered this should not have a big impact and we 
	chose $b=0.12$ and $d=0.02$ for the differentiated as well as the dedifferentiated cells. Landsberg et al. 
	observed that the growth kinetics appear to be the same for both cell types, see 
	Supplementary Figure 11 in \cite{Landsberg:2012vn}. \\

	\item \textit{Choice of the competition:} We assume that the competition 
	has a very little effect here because the tumor grows exponentially in the observed time frame and does 
	not come close to its equilibrium. We choose the competition between melanoma cells of the same type as $c(x,x)=c(y,y)=0.00005$ and between different types of melanoma cells as $c(x,y)=c(y,x)=0.00002$. The values are not set to $0$ since the melanoma can grow only up to a finite 
	size. \\
	
	\item \textit{Choice of the switch parameters: } We can now estimate the
	switching parameters by using the data of Supplementary
	Figure 9e in \cite{Landsberg:2012vn}. In this  experiment where cell division is inhibited, we
	can set $b=0$.  Furthermore, the amount of
	TNF-$\alpha$ is constant and we set here $\frakn_w=2$. Thus, the dynamics of the melanoma populations is described by
	\begin{equation}
	\begin{array}{ll}
	\dot \frakn_{x} &= \frakn_{x} \big(-d(x)-c(x,x) \frakn_{x}-c({x
		,  y}) \frakn_y -2s_w(x,y)-s(x,y)\big)+s({y,x}) \frakn_y \\
	\dot \frakn_y &= \frakn_y \big(-d(y)-c({y,y}) \frakn_y-c({y, x}) \frakn_{x}-s({y,x})\big)
	+\left(2s_w(x,y)  +s({x, y})\right) \frakn_{x} \\
	\end{array}
	\end{equation}
	
	At the beginning of their observations the
	switch is very slow and speeds up after the first 24 hours. We assume that there is a delay until the 
	reaction really starts and thus we choose the proportions at day 1 ($\frakn_{x}=0.81$ and $\frakn_y=0.19$) as initial data and choose switching 
	parameters such that roughly the concentrations at day 2  ($\frakn_{x}=0.45$ and $\frakn_y=0.54$) and 3  ($\frakn_{x}=0.24$ and $\frakn_y=0.72$) are reached as shown in Figure 
	\ref{Switch}. Thereby we obtain
	$s(x,y)=0.0008, s(y,x)=0.065$ and $s_w(x,y)=0.33$. 
	Note that the experiments we refer to provide only in vitro data and it is not clear if the in vivo situation is similar.\\
	
	\begin{figure}[h!]
		\centering
		\includegraphics[width=0.49\textwidth]{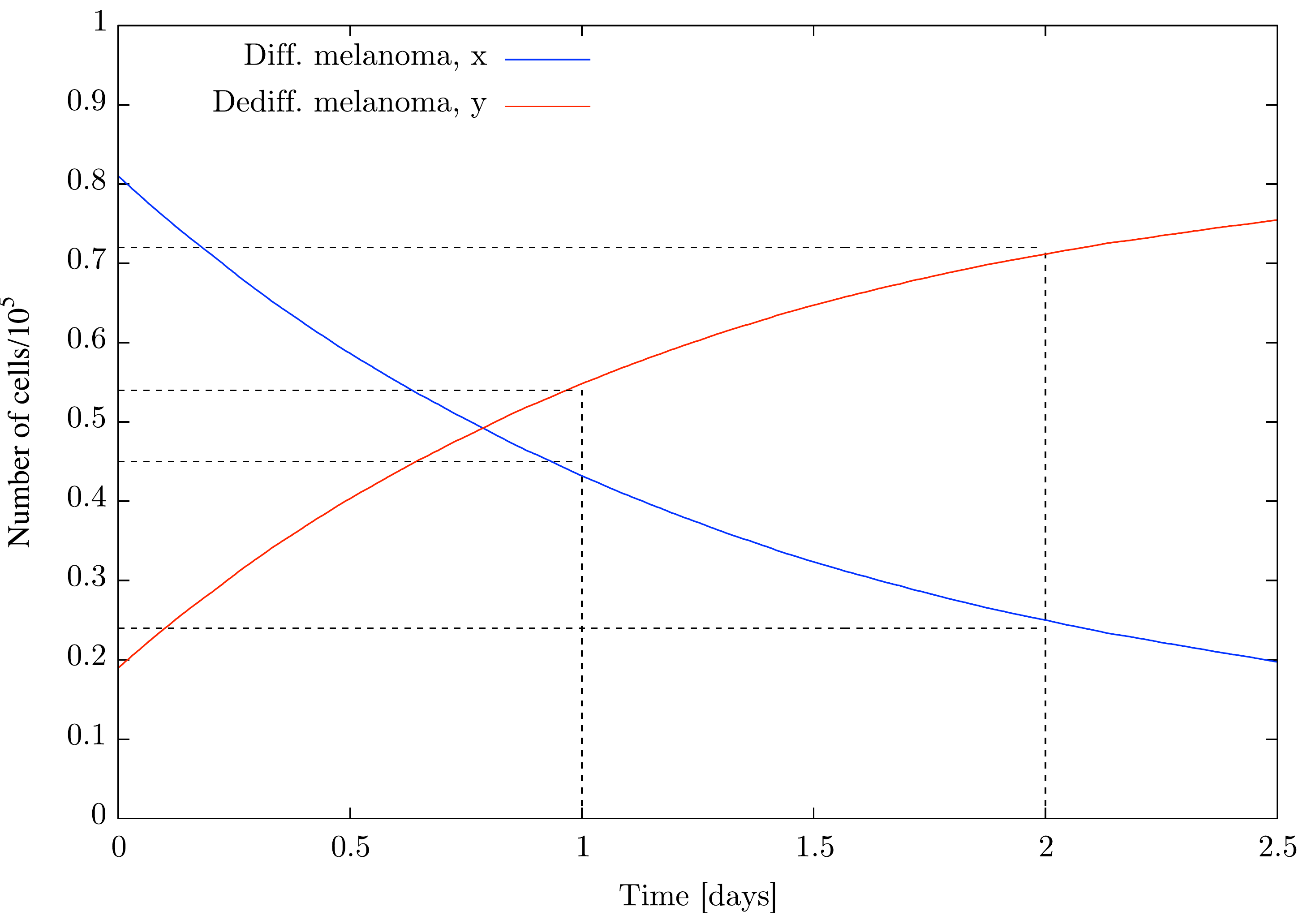}
		\captionof{figure}{\small
			Switch in the in vitro experiments for inhibited
			cell division and constant concentration of
			TNF-$\alpha$. Dashed lines indicate experimental data.} \label{Switch} 
	\end{figure}

	\item \textit{Choice of parameters concerning T-cells: }
	It remains to characterize the T-cells. Their natural birth rate is set to $0$ since they are transferred by 
	adoptive cell transfer and not produced by the mice themselves and do not proliferate in absence of 
	targets. We assume that they have a relatively high birth rate 
	depending on the amount of cancer cells present, $b(z_x,x)=b(z_y,y)=2$ and produce one TNF-$\alpha$ molecule when 
	they divide, $\ell^{\text{prod}}_w(z_x,x)=\ell^{\text{prod}}_w(z_y,y)=1$. Furthermore, we assume that  $4.5$  cancer cells can be killed per hour (including indirect mechanisms), 
	$t(z_x,x)=t(z_y,y)=108$.
	The rate of death for the T-cell population is chosen as  $d(z_x)=d(z_y)=0.12$. These parameters are chosen such that 
	the qualitative behavior 
	of the tumor was recovered. We choose the same parameters for the 
	second T-cell type as for the first one because there are no data concerning the second T-cell type. \\
	
	\item \textit{Choice of starting values and the scale $K$: } We
	set $K=10^5$, the initial value for the differentiated melanoma
	cell population to $1$  and to $0$ for the population of dedifferentiated melanoma cells. The ratio of differentiated and dedifferentiated cells is not known for small tumors which do not result from cell transfer of cells of in vitro cell lines.
	The initial value of the T-cell population is set to $0.02$. 
	We assume that the T-cells appear directly in the tumor, i.e.\ the migration phase into the tumor is not modeled.\\
\end{itemize}

To sum up, biological rates (per day) and initial conditions (in 100 000 cells) are:
\begin{equation}\label{1tcell-bio-parameter}
\scriptsize
\begin{array}{l@{\hspace{1cm}}l@{\hspace{1cm}}l@{\hspace{1cm}}l}
b({x})=0.12					&		b(y)=0.12			&	b(z_x,{x})=2			&	\ell^{\text{prod}}_w(z_x,x)=1\\
d({x})=0.02					&		d(y)=0.02			&	t(z_x,{x})=108			&	\ell^{\text{kill}}_w(z_x,x)=0 \\
c({x},{x})=5\cdot 10^{-5}	&	c({x},y)=2\cdot 10^{-5}	&	d(z_x)=0.12				&	d(w)=0.2\\
c(y,{x})=2\cdot 10^{-5}		&	c(y,y)=5\cdot 10^{-5}	&							&	 s_w(x,y)=0.33\\
s(x,y)=0.0008				&		s(y,x)=0.065		&							&\\
\frakn_{{x}}(0)=1			&	\frakn_{y}(0)=0			&	\frakn_{z_x}(0)=0.02	&	K=10^5
\end{array}
\end{equation}

The additional parameters in the case where a second T-cell is used are: 
\begin{equation}\label{2tcell-bio-parameter}
\begin{array}{lll}
t(z_y,{y})=108&\ell^{\text{prod}}_w(z_y,y)=1 & d(z_y)=0.12 \\
b(z_y,{y})=2 & \ell^{\text{kill}}_w(z_y,y)=0				&\frakn_{z_y}(0)=0.02
\end{array}
\end{equation}

\section{Arrival of a mutant}\label{sec-mutant}

 The model we defined in Section \ref{sec-model} can be seen as a generalization of one of the standard models of adaptive dynamics, usually called BPDL, which was introduced in publications of Bolker, Pacala,
 Dieckmann and Law \cite {BolPac1,BolPac2,DieLaw}. In a possibly continuous trait space $\calX\subseteq\R^d$, the BPDL model allows for each individual  to reproduce, with or without mutation, or die due to natural death or to competition (e.g.\ this amounts to only keep parameters $b(x), d(x), c(x,y), m(x,y)$ and $\mu$ in Subsection \ref{notations}).
 %

	 In the limit of large populations ($ K \to \infty $ and $ \mu$ fixed)
	Fournier and
	M\'el\'eard \cite {FouMel2004} proved a law of large numbers: the process
	converges to a system of
	deterministic integro-differential equations. 
	In the  case $\mu\equiv0$ the process converges
	to the solution of a system of coupled logistic equations (of
	Lotka-Volterra type) without mutations.
	The limit of large populations ($ K \to \infty $) followed by the limit of rare mutations ($ \mu \to0 $) on
	the timescale $ t \sim \log (1 / \mu) $ 
	was studied by Bovier and Wang \cite {BovWan2013} and a deterministic jump process
	is obtained as limiting behavior.
	
	 The simultaneous  limits of large populations and rare mutations, where  $ (K, \mu) \to (\infty, 0) $ at a rate such that
	$ 1 / (K \mu) \gg \log K $ and a timescale $ t \sim 1 / (K \mu) $ 
	was studied by Champagnat and M\'el\'eard
	\cite {Cha2006,ChaMel2011}. At this scale the system has time
	to equilibrate between two successive mutations. The long-term behavior of
	the population can be described  as a  Markov jump process along a sequence of 
	equilibria of, in general, polymorphic populations.
	An important (and in some sense generic) special case occurs when the mutant population  fixates while resident 
	population dies out in each step. The corresponding jump process is called  the \emph {Trait Substitution
		Sequence} (TSS) in adaptive dynamics.  Champagnat \cite{Cha2006} derived criteria in the context of individual-based 
	models under which convergence to the  TSS can be proven.
	 The general process is called the \emph {Polymorphic  Evolution Sequence}
	(PES). It is described partly implicitly in \cite{ChaMel2011}, as it involves the identification of attractive fixed points in a sequence of Lotka-Volterra equations that
	are in general not tractable analytically.
	
	 In situations when the population converges to a TSS, one may take a further limit of
	 small mutation steps,  $\sigma $, and look at the time scale $t\sim1/\sigma^2$.  The mutant is then of trait $y=x+\sigma h$ where $h$ is
	 chosen according to a probability kernel on $\calX$.  For birth and death rates that vary  smoothly on $\calX$, this
	 ensures a vanishing evolutionary advantage  of mutants  when $\sigma \to 0$. 
	The TSS converges in this limit to the so-called 
	\emph{Canonical Equation of Adaptive Dynamics} (CEAD), see 
	e.g.\ \cite{ChaMel2011},
	which describes the continuous evolution of a monomorphic population in a fitness landscape.  
	 The convergence of the individual-based model to the CEAD in a single step has recently  be shown  by 
	Baar et al.\  \cite{BBC2015}, i.e.\ the limit  $(K, \mu, \sigma) \to (\infty, 0,0) $ of large populations, rare mutations and small mutation steps are taken simultaneously, provided 
	$ 1/ (K \mu) \gg \log (K)/\sigma  $ and $ 1 / \sqrt {K} \ll \sigma \ll1 $. 
	The time-scale on which this convergence takes place is  $ t \sim 1 / (K \mu \sigma^2) $, and corresponds to the combination of the previous two.
	
	 Costa et al.\ 
	\cite{CHLM} study an extension of the model with a predator-prey relation. The predator-prey kernel is an  explicit function of parameters describing defense strategies for preys, together with the ability of predators to
	circumvent the defense mechanism. In the limit  $ 1 / (K \mu) \gg \log K $, convergence to a Markov jump process generalizing
	the \emph{Polymorphic Evolution Sequence} is derived, and in the subsequent limit $ \sigma \to0 $, in the case of a monomorphic prey and  predator populations, convergence to an extended version of the CEAD is obtained.
	

\subsection{Rare mutations and fast switches}\label{sec-PES}

In this section, we give  a  generalization of the Polymorphic Evolution Sequence in the case of fast switches in the phenotypic space, without therapy (no predator-prey term). We consider the case of rare mutations in large populations on a timescale such that a population reaches equilibrium before a new mutant appears:
\begin{equation}\label{conditions}
\forall V>0, \qquad \exp(-VK)\ll \mu^K_g \ll \frac{1}{K\log K}, \qquad \text{as } K\to\infty.
\end{equation}

Champagnat and M\'el\'eard's proof of convergence to the PES \cite{ChaMel2011} relies on a precise study of the way a mutant population fixates, which we now describe. A crucial assumption is that the large population limit is a competitive Lotka-Volterra system with a unique stable\footnote{By stable fixed point we mean that the eigenvalues of the Jacobian matrix of the system at the fixed point have all strictly negative real parts.} fixed point $ \bar \frakn$.
The main task is to study the invasion of a mutant  that has just appeared in a
population close to equilibrium. The invasion can be divided into three steps. First, as long as the
mutant population size is smaller than $K \epsilon$ for a fixed small $\epsilon>0$, the resident population
stays close to its equilibrium. Therefore the mutant population  can be approximated by a binary branching process. Second, once the mutant population reaches the level $K \epsilon$, the whole system
is close to the solution of the corresponding deterministic system (this is a consequence of
Proposition~\ref{det-limit}) and reaches an $\epsilon$-neighborhood of $ \bar \frakn$ in finite time. Finally, the subpopulations which have a zero coordinate in $ \bar \frakn$ can be approximated  by sub-critical branching processes. The durations of the first and third steps are proportional to $\log(K)$, whereas that of
the second step is bounded.  The second inequality in \eqref{conditions} guarantees that, with high
probability, the three steps of invasion are completed before a new mutation occurs.

\subsubsection{Invasion fitness}

  Given a population in a stable equilibrium  that populates a certain set
  of traits, say $M\subset \calX$, the invasion fitness $f(x,M)$ is the growth rate of a population consisting of a single individual 
  with trait $x\not\in M$ in the presence of the equilibrium population $ \bar\frakn$ on $M$. 
  In the case of the BPDL model, it is given 
  by 
  \begin{equation}
  f(x,M)= b(x)-d(x)-\sum_{y\in M} c(x,y) \bar\frakn_y.
  \end{equation}
  Positive $f(x,M)$ implies that a mutant appearing with trait $x$ from the equilibrium 
  population on $M$ has a 
  positive probability (uniformly in $K$) to grow to a population of size of order $K$; 
  negative invasion fitness implies that such a mutant 
  population will die out with probability tending to one (as $K\to\infty$) before this happens.
  
  Invasion fitness is a fundamental concept in the analysis of stochastic population
  models. We first generalize it to the case where fast phenotypic switches are present, for pure tumor evolutions, i.e. we  ignore the T-cells and the cytokines in our model.

Let us consider an initial population of genotype $g$ (associated with $\ell$
different phenotypes $ p_1,\ldots, p_\ell$) which is able to
mutate at rate $\mu_g^K$ to another genotype $ g '$, associated with $k$
different phenotypes $ p'_1,\ldots, p'_k$. The assumption \eqref{conditions} ensures that no mutation occurs during the equilibration phase in the phenotypic space.

Consider as initial condition $\frakn(0)=(\frakn_{(g,p_1)}(0),\ldots,\frakn_{(g,p_\ell)}(0))$ a stable fixed point, $\bar\frakn$, of the following system:
\begin{equation}\label{deterministic-system-initial}
\dot \frakn_{(g,p_i)} =\frakn_{(g,p_i)} \left(b_i-d_i-\sum_{j=1}^{\ell}c_{ij}\frakn_{(g,p_j)}-\sum_{j=1}^{\ell}s_{ij}\right)
+\sum_{j=1}^{\ell}s_{ji}\frakn_{(g,p_j)},\quad\quad i=1,\ldots, \ell.
\end{equation}
We write for simplicity $b_i=b(p_i)$, $d_i=d(p_i)$, $c_{ij}=c(p_i,p_j)$, $s_{ij}=s_g(p_i,p_j)$, and later
$b'_i=b(p'_i)$, $d'_i=d(p'_i)$, $c'_{ij}=c(p'_i,p'_j)$, $\tilde c_{ij}=c(p_i,p'_j)$, $s'_{ij}=s_{g'}(p'_i,p'_j)$.

We want to analyze whether a  single mutant appearing with a new genotype $g'$ (and one of its  possible 
phenotypes $p'_i$), has a positive probability to give rise to a population of size of order $K$.  
Using the same arguments as  Champagnat et al. \cite{Cha2006,ChaMel2011}, it is easy to show that as long as the mutant population 
has less than $\epsilon K$ individuals (with $\epsilon\ll 1$), the mutant population  
$(g ', p'_1),\ldots, (g', p'_k) $ is well approximated  by a $k$-type
branching process, where each individual undergoes
binary branching, death, or switch to another  phenotype with the following rates:
\begin {equation}\label{MTBP}
\left.
\begin {array}{l@{\quad\text {with  rate}\quad}l}
p'_i  \to p'_i p'_i  &  b'_i\\
p'_i  \to \varnothing & d'_i + \sum_{l=1}^{\ell}\tilde c_ {il} \bar \frakn_l \\
p'_i  \to p'_j &   s'_ {ij}\\
\end {array}\right\}\quad \text{ for } i,j\in\{1,\ldots,k\}.
\end {equation}
where $(\bar \frakn_1,\ldots,\bar \frakn_\ell)$ denotes the fixed point of \eqref{deterministic-system-initial}. 
We will assume that the switch rates $s'_{ij}$
are the transition rates of an irreducible Markov chain on $\{1,\dots, k\}$. The simplest example is the case where $s'_{ij}>0$, for all $i,j\in\{1,\ldots,k\}$.

Multi-type branching processes have 
been analyzed by  Kesten and Stigum \cite{KesSte1,KesSte2,KesSte3}  and Athreya and Ney \cite{AthNey}. 
Their behavior are classified in terms of the  matrix $A$, given by
\begin{equation}
A=\begin{pmatrix}
f_1		&	s'_{12}	&	\ldots	&	s'_{1k}	\\
s'_{21}	&	f_2		&			&	\vdots	\\
\vdots	&			&	\ddots	&			\\
s'_{k1}	&	\ldots	&			& f_k
\end{pmatrix},
\end{equation}
where  
\begin{align}\label{apparent-fitness}
f_i:&=b'_i-d'_i - \sum_{l=1}^{\ell}\tilde c_ {il} \cdot \bar \frakn_{l}-\sum_{j=1}^{k}s'_{ij},\quad  i=1,\ldots,k. 
\end{align}
Note that $f_i$ \emph{would be} the invasion fitness of phenotype $i$ \emph{if} there was no switch back from the
other phenotypes to $i$ (or if all switched cells would be killed). 
It is well-known that 
the multi-type process is super-critical, if and only if the largest eigenvalue, $\l_1=\l_1(A)$, of the matrix
$A$ is strictly positive, meaning that 
if $\l_1>0$, the mutant population will grow (with rate $\lambda_1$)  to any desired 
population size before dying out. Thus $\l_1(A)$ is the appropriate generalization of the invasion fitness of a genotype:
\begin{equation}
\label{invasion-fitness-case1}
F(g',g):=\l_1(A).
\end{equation}
This notion can easily be generalized to the case when the initial condition is the equilibrium population of several 
coexisting genotypes.
Note that this notion of invasion fitness of course  reduces to  the  standard one of \cite{Cha2006} if there is 
only  one mutant phenotype.
This settles the first step of our analysis, which is the invasion of
the mutant. \\

\subsubsection{Towards a generalized Polymorphic Evolution Sequence}

In fact, one can say more about how the mutant population grows. 
Write $Z^{(i)}_j(t)$ for the number of individuals of phenotype $p_j$
existing at time $t$ for this branching process when the first mutant
is of phenotype $p_i$. 
Then, for $i,j\leq k$, 
\begin{equation}
\mathbb E(Z^{(i)}_j(t))=[M(t)]_{i,j}
\end{equation}
where $M(t)$ is the $k\times k$-matrix 
\begin{equation}
M(t)=\exp(At).
\end{equation}

Assume that the largest eigenvalue $\l_1(A)$ is simple and strictly positive.
Let $v$ and $u$ be the left and right eigenvectors of $A$  associated to $\l_1$, normalized such that 
$u\cdot v=1$ and $u\cdot 1=1$.
The extinction probability vector $q=(q_1,\ldots,q_k)$ where
$q_i=\mathbb {P}(Z^{(i)}(t)=0\text{ for some } t)$ is the unique solution of the
system of equations:
\begin{equation}
d'_{i}+\sum_{l=1}^{\ell}\tilde c_{il}\bar{\frakn}_l+b'_i q_i^2+\sum_{j=1}^{k}s'_{ij} q_j=q_i\left(d'_{i}+\sum_{l=1}^{\ell}\tilde c_{il}\bar{\frakn}_l+b'_i +\sum_{j=1}^{k}s'_{ij} \right),\quad\quad i=1,\ldots,k
\end{equation}
which has in general no analytical solution. Then the following 
limit theorem holds \cite{KesSte1}:
\begin{equation}
\label{limit-thm}
\lim_{t\to\infty}\left(Z_1^{(i)}(t),\ldots,Z_k^{(i)}(t)\right)e^{-\lambda'_1t}=W_i\cdot (v_1,\ldots,v_k) \quad \text{a.s.}
\end{equation}
where $(W_i)_{ i=1,\dots,k}$ is  random vector with non-negative entries such that 
\begin{equation}\label{limit-thm-bis}
\mathbb P(W_i=0)=q_i\quad\text{and}\quad \mathbb E (W_i)=u_i.
\end{equation}
In particular, conditionally on survival,
the phenotypic distribution of the mutant populations converges
almost surely to a deterministic quantity, which  moreover does not depend on
the phenotype of the initial mutant, namely
\begin{equation}
\label{limit-ratio}
\lim_{t\to\infty}\frac{Z^{(i)}_{j'}(t)}{\sum_{j=1}^{k}Z^{(i)}_j(t)}=\frac{v_{j'}}{\sum_{j=1}^{k}v_j},
\quad\forall i,j'=1,\ldots,k.
\end{equation}
For us, this implies the important fact  that the phenotypic structure of the mutant population once it reaches the level
$\ep K>0$ is almost deterministic. 
Then, conditionally
on survival, \eqref{limit-thm} implies that the time  $\tau_{\ep K}$ until the total mutant population reached $\ep K$ 
is of order $\log(K)$. Moreover, the
proportions of the $k$ types of phenotypes converge to 
deterministic quantities given above, 
\begin{equation}
\frac1K (Z_1(\tau_{\ep K}),\ldots,Z_k(\tau_{\ep K}))\to\left(\frac{\ep
	v_1}{\sum_{j=1}^{k}v_j},\ldots,\frac{\ep v_k}{\sum_{j=1}^{k}v_j}\right), \quad \text{in distribution, as }K\to\infty.
\end{equation}
Once the mutant population has reached the level $\ep K$, the behavior of the process can be
approximated by the solution of the deterministic system:
\begin{align}\label{deterministic-system-mutant}\nonumber
\dot \frakn_{(g,p_i)} &=\frakn_{(g,p_i)} 
\left(b_i-d_i-\sum_{j=1}^{\ell}s_{ij}
-\sum_{j=1}^{\ell}c_{ij}\frakn_{(g,p_j)}
-\sum_{j=1}^{k}\tilde c_{ij}\frakn_{(g', p'_j)}\right)
+\sum_{j=1}^{\ell}s_{ji}\frakn_{(g,p_j)},\\
\dot \frakn_{(g', p'_i)} &=\frakn_{(g', p'_i)} 
\left(b'_i-d'_i-\sum_{j=1}^{k}s'_{ij}
-\sum_{j=1}^{k}c'_{ij}\frakn_{(g', p'_j)}
-\sum_{j=1}^{\ell}\tilde c_{ij}\frakn_{(g,p_j)}\right)
+\sum_{j=1}^{k}s'_{ji}\frakn_{(g',p'_j)},
\end{align}
(where $i$ runs from 1 to $\ell$ in the first line and from 1 to $k$ in the second one)
with initial conditions in a small 
neighborhood of 
\begin{equation}\label{start-point}
(\frakn_{(g,p_1)}(0),\ldots,\frakn_{(g,p_\ell)}(0),\frakn_{(g',p'_1)}(0),\ldots,\frakn_{(g',p'_k)}(0))=
\left(\bar{\frakn}_1,\ldots,\bar{\frakn}_\ell,
\frac{\ep
	v_1}{\sum_{j=1}^{k}v_j},\ldots,\frac{\ep v_k}{\sum_{j=1}^{k}v_j}\right).
\end{equation}
If the system \eqref{deterministic-system-mutant} is such that a neighborhood of 
\eqref{start-point} is  attracted to the same stable fixed point, we are in the same situation as in 
Champagnat and M\'el\'eard \cite{ChaMel2011} and get the generalization of the Polymorphic Evolution Sequence on the genotypic trait space. 

Observe that in the case when the system of equations \eqref{deterministic-system-mutant}
has multiple attractors, and different points near \eqref{start-point} lie in different basins of attraction, then for finite 
$K$, the choice of attractor the system approaches may be random. 
The characterization of the asymptotic behavior of the system
\eqref{deterministic-system-mutant} is needed to describe the final
state of our stochastic process. This is in general a very difficult
and complex problem, which is not doable analytically and  requires numerical analysis.

Figure  \ref{mutants} shows examples where in a population consisting only of type $(g,p)$ a mutation to genotype $g'$ occurs. $g'$ is associated to two possible phenotypes $p_1$ and $p_2$. 

%

\begin{figure}[h!]
	\centering
	\begin{minipage}{.49\textwidth} 
		A \includegraphics[width=\textwidth]{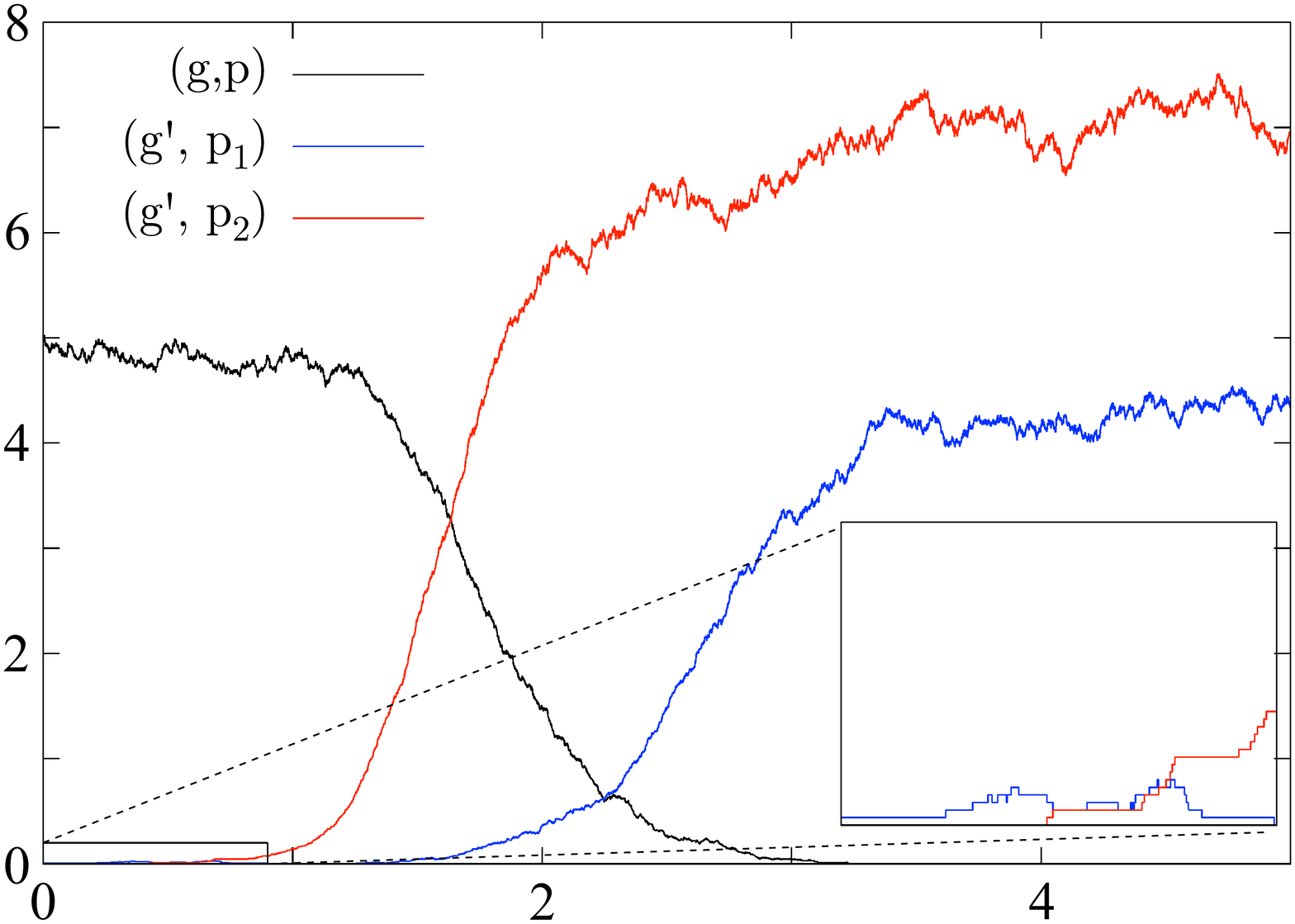}
	\end{minipage}\hfill
	\begin{minipage}{.49\textwidth} 
		B \includegraphics[width=\textwidth]{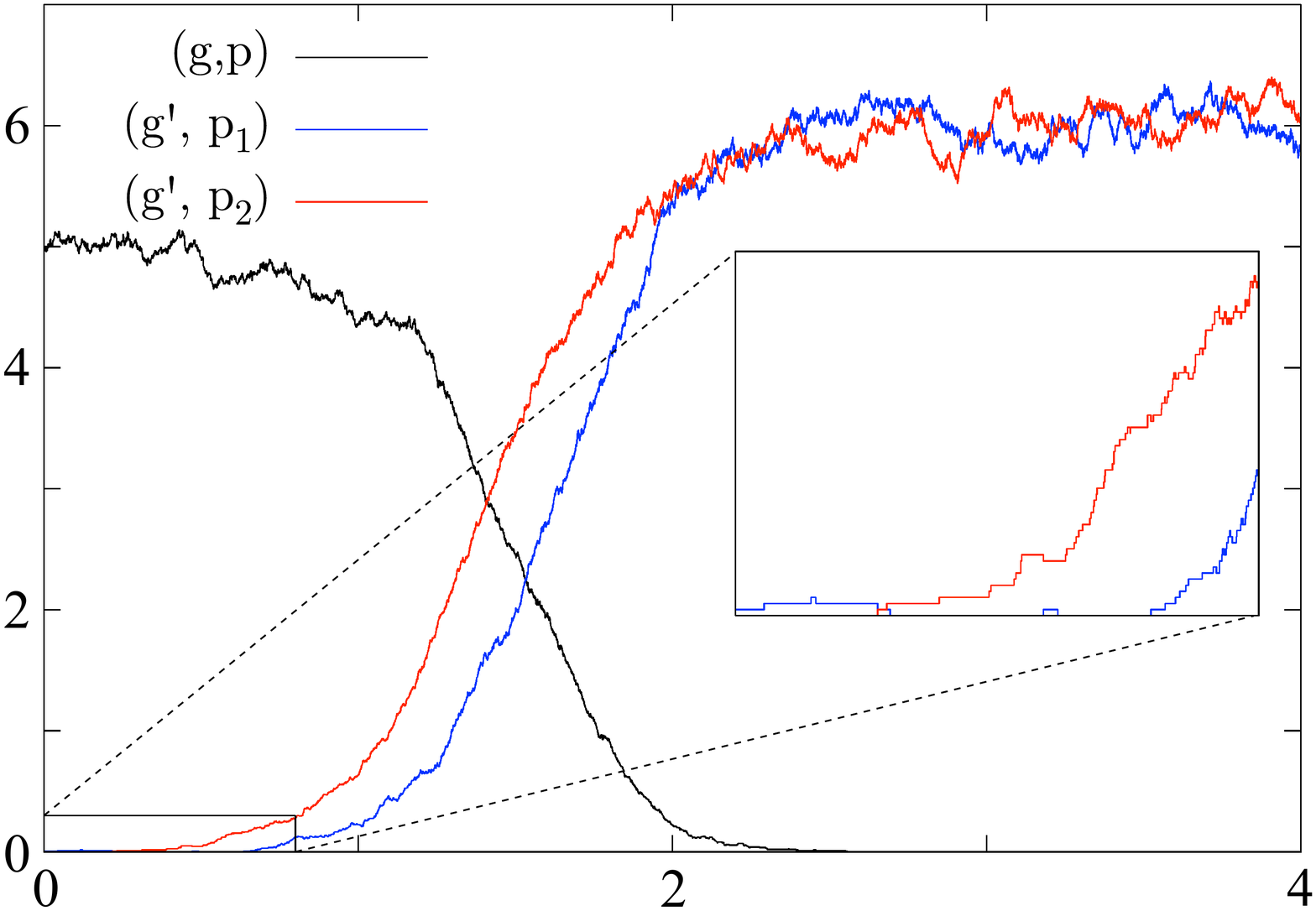}
	\end{minipage}
	\caption{\small { Simulations for rare mutations in combination with fast switching, where the number of individuals 
			divided by $200$ is plotted versus time.
			(A) The mutant phenotype $p_2$ has a negative initial growth rate but can switch to $p_1$  which has a positive one. The fitness of the genotype $g'$ is positive.
			(B) The fitness of the mutant genotype $g'$ is positive, although 
			each phenotype has a 
			negative initial growth rate. This is possible because an outgoing switch is a loss of a particle for a phenotype, but not for the 
			whole genotype. 
			%
			\label{mutants}} }
\end{figure}

Figure \ref{mutants} (A) is realized with following parameters:
\begin{equation}\label{nobackswitch-parameter}
\begin{array}{l@{\hspace{1cm}}l@{\hspace{1cm}}l@{\hspace{1cm}}l}
b_0=5			&	b_1=6				&	b_2=6						&	s_{12}=0.1\\
d_0=0			&	d_1=0				&	d_2=0						&	s_{21}=2\\
c_{00}=1		&	c_{10}=1			&	c_{20}=1					&	\\
c_{01}=1		&	c_{11}=1			&	c_{21}=0					&\\
c_{02}=0		&	c_{12}=0			&	c_{22}=1					&\\
\frakn_{(g,p)}(0)=5&\frakn_{(g',p'_1)}(0)=0&\frakn_{(g',p'_2)}(0)=1	&	K=200
\end{array}
\end{equation}
In this case, $p'_2$ can switch to $p'_1$ but the back-switch is very weak, and we have $f_2<0$ and $f_1>0$ according to definition \eqref{apparent-fitness}. 
The global fitness of the genotype $g'$ is positive and close to $f_1$. The coexistence of the two phenotypes depends on the cross-competition $c_{12}$ and $c_{21}$. 

Figure \ref{mutants} (B)  shows an example of the case discussed above with the following parameters:
\begin{equation}\label{helix-parameter}
\begin{array}{l@{\hspace{1cm}}l@{\hspace{1cm}}l@{\hspace{1cm}}l}
b_0=5				&	b_1=6					&	b_2=6					&	s_{12}=2\\
d_0=0				&	d_1=0					&	d_2=0					&	s_{21}=2\\
c_{00}=1			&	c_{10}=1				&	c_{20}=1				&	\\
c_{01}=1			&	c_{11}=1				&	c_{21}=0				&\\
c_{02}=0			&	c_{12}=0				&	c_{22}=1				&\\
\frakn_{(g,p)}(0)=5	&\frakn_{(g',p'_1)}(0)=1	&	\frakn_{(g',p'_2)}(0)=0	&	K=200
\end{array}
\end{equation}
For these parameters $f_1$ and $f_2$ as defined in \eqref{apparent-fitness} are negative, but, due to the 
cooperation of the two phenotypes, the fitness of the genotype is positive and it 
invades with positive probability as indicated by  the definition \eqref{invasion-fitness-case1}. 
Moreover, both phenotypes appear on a macroscopic level.

\subsection{Interplay of mutation and therapy}\label{sec-brc}

In the previous section we considered the probability of invasion of a mutant when the 
resident population is at an equilibrium given by a stable fixed point.  In the context of therapy, 
there are phases when populations shrink and regrow due to treatment and relapse phenomena. 
In the medical literature, there are frequent allusions to the possibility that such growth cycles 
may induce fixation of a ``super-resistant mutant", see e.g. \cite{FrankRosner12,GilVerGat,Greaves2012}.
It is important to understand whether and under what circumstances such effects may happen. 

Here we show an example where the appearance of a mutant genotype may indeed
be enhanced  under treatment. 
We consider \emph{birth-reducing competition} (BRC) between tumor cells.
In such a case, a large population at equilibrium may encounter fewer births and hence mutations, than a smaller population growing towards its equilibrium size.

Let us discuss in more detail how the birth-reducing competition can have a crucial
effect on the mutation timescale. For the sake of simplicity we consider an example where  the switching rates are 
set to $0$. Consider a melanoma population
$(g,p)$ which is able to mutate to a fitter type of melanoma
$(g',p')$. We allow for one T-cell population attacking the resident melanoma population since this is the simplest scenario where the effect of therapy in this context can be explained. As the presence of TNF-$\alpha$ only influences the switch between
phenotypes, it does not play any role in this example and we can
set the corresponding parameters ($\ell_w$ and $d(w)$) to zero.
If $\mu_g^K\to 0$ as $K\to \infty$
the limiting deterministic system is given by:


\begin{align}\label{mut-therapy-system2}
\dot \frakn_{(g,p)} &= \frakn_{{(g,p)}}\left(
b({p})-d({p})-c_b({p},{p}) \frakn_{{(g,p)}}-c_b({p},p')
\frakn_{(g',p')} -t(z,p) \frakn_z\right)\nonumber \\
\dot \frakn_{{(g',p')}} &= \frakn_{{(g',p')}}  \left( b({p'})-d(p')-c_b({{p'},{p'}}) \frakn_{{(g',p')}}-c_b({{p'},p}) \frakn_{{(g,p)}}\right) \nonumber  \\
\dot \frakn_z &=  \frakn_z (b(z,p) \frakn_{{(g,p)}}-d(z))
\end{align}
Note that the mutation term does not appear in the deterministic system and that the difference between birth-reducing competition and usual competition is not visible. The effects we are looking for are intrinsically stochastic and
happen on time-scales that diverge with $K$.

If the
competition is only of
birth-reducing type, then  the total mutation rate of the population of type  $(g,p)$ at time $t$ is
\begin{equation}
\mathfrak m(  \nu_t^K(g,p)):=\mu^K_g\left\lfloor b(p)-c_b(p,p) \nu_t^K(g,p)\right\rfloor_+  \nu_t^K(g,p)K.
\end{equation}

This is a positive and concave function of $ \nu_t^K(g,p)$ on the interval $[0,b(p)/c_b(p,p)]$, see Figure \ref{parabola}.
If the population is at equilibrium (without or before therapy),
meaning $ \nu_t^K(g,p)=\bar \frakn_{(g,p)}=(b(p)-d(p))/c_b(p,p)$, then the time until a mutation occurs is
approximately exponentially distributed with parameter equal to
\begin{equation}
\mu_g^K K \cdot \left(b(p)-{c_b(p,p)}\bar  \frakn_{{(g,p)}}\right)
\bar \frakn_{{(g,p)}} =
\mu^K_g K \cdot d(p) \bar \frakn_{{(g,p)}} .
\end{equation}
If $d(p)$ is smaller than $b(p)/2$, then $\bar \frakn_{{(g,p)}}$ is bigger than $b(p)/2c_b(p,p)$ and $\frakm(\bar{\frakn}_{(g,p)})$ is not maximal. Smaller 
populations, more precisely in between $d(p)/c_b(p,p)$ and $ \bar \frakn_{(g,p)}$,
have a higher total mutation rate. 

The interesting scaling of the mutation rate is \mbox{$K\mu_g^K\to\alpha>0$} as $K\to\infty$.
In this case, there is 
a number of mutations of order one while the population grows by $\mathcal O(K)$ individuals. For lower mutation rates, 
no mutant can be expected before the population reaches its equilibrium, while for higher rates, mutations occur unrealistically fast. Since for $\mu_g^K\to 0$ the mutation term does not appear in the deterministic 
system, the difference between BRC and usual competition is invisible. 

 During therapy, a tumor which is close to equilibrium (similar to Figure  \ref{IMT} (B)) can shrink to a small size (similar to Figure  \ref{IMT} (A)): the introduction of T-cells in the system
 lowers the population size of melanoma, and the total mutation
 rate in the tumor population of type ${(g,p)}$ can be larger during treatment, see Figure  \ref{IMT} (C). 
 This means that treatment could lead to earlier mutations and thereby accelerate the evolution towards more aggressive tumor variants.

\begin{figure}[h!]
	\centering
	\includegraphics[width=.45\textwidth]{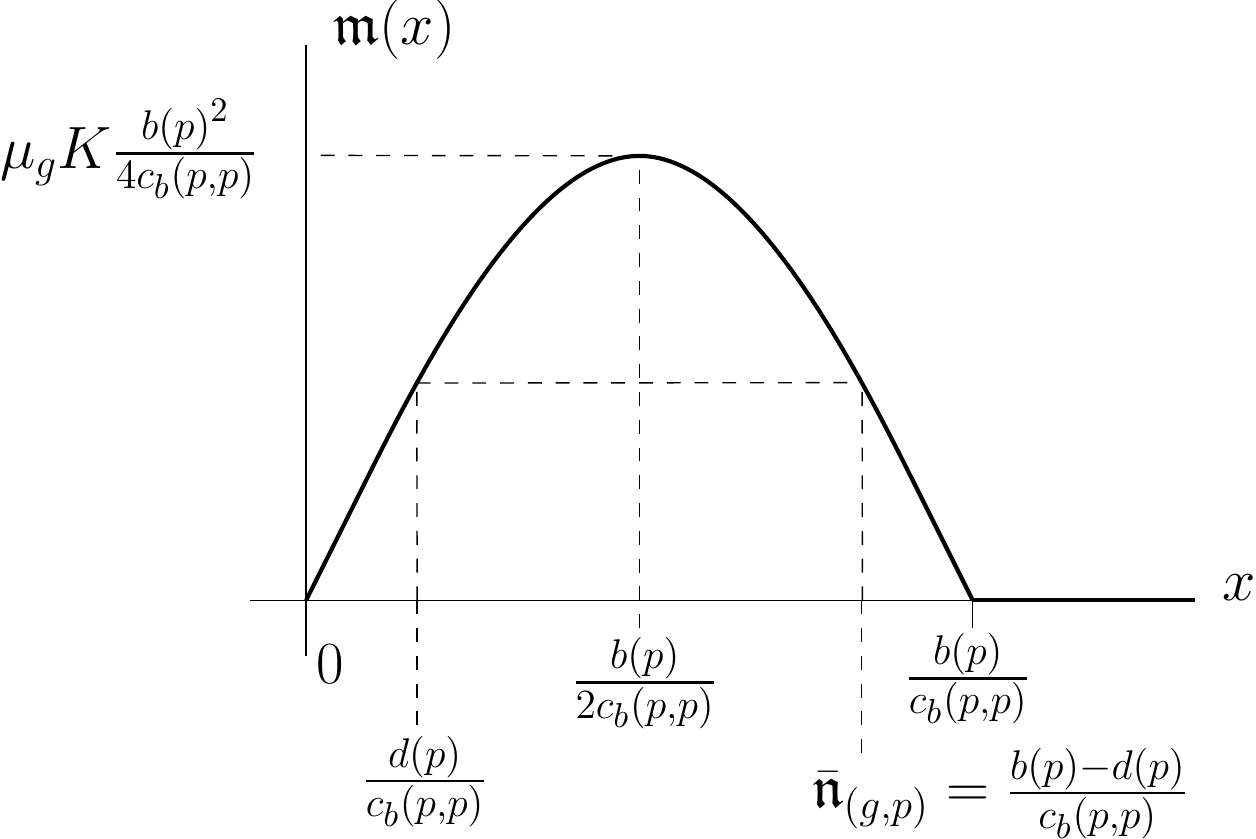}
	\caption
	{Initial total mutation rate of the population $(g,p)$.}
	\label{parabola}
\end{figure}

The simulations are obtained with the following parameters:
\begin{equation}\label{IMT-parameter}
\begin{array}{l@{\hspace{1cm}}l@{\hspace{1cm}}l@{\hspace{1cm}}l}
b(p)=4					&	b(p')=6					&	b(z,p)=20				&	m({(g,p)},{(g',p')})=1\\
d(p)=0.1				&	d(p')=1					&	t(z,p)=10				&	\mu_g=10^{-3}\\
c_b(p,p)=3				&	c_b(p',p)=0.8			&	d(z)=6					&	K=10^3\\
c_b(p,p')=0.8			&	c_b(p',p')=1			&							&\\
\frakn_{{(g,p)}}(0)=1.3	&\frakn_{{(g',p')}}(0)=0	&\frakn_z(0)=0\text{ or }0.1&
\end{array}
\end{equation}

Note that this example provides an interesting situation of interplay
between therapy and mutation. By lowering the melanoma population, the
T-cell therapy actually increases the probability for it to mutate to
a potentially fitter and pathogenic genotype, which is not affected by the T-cells.

\begin{figure}[h!]
	\centering
		\begin{minipage}{.49\textwidth} 
			A \includegraphics[width=\textwidth]{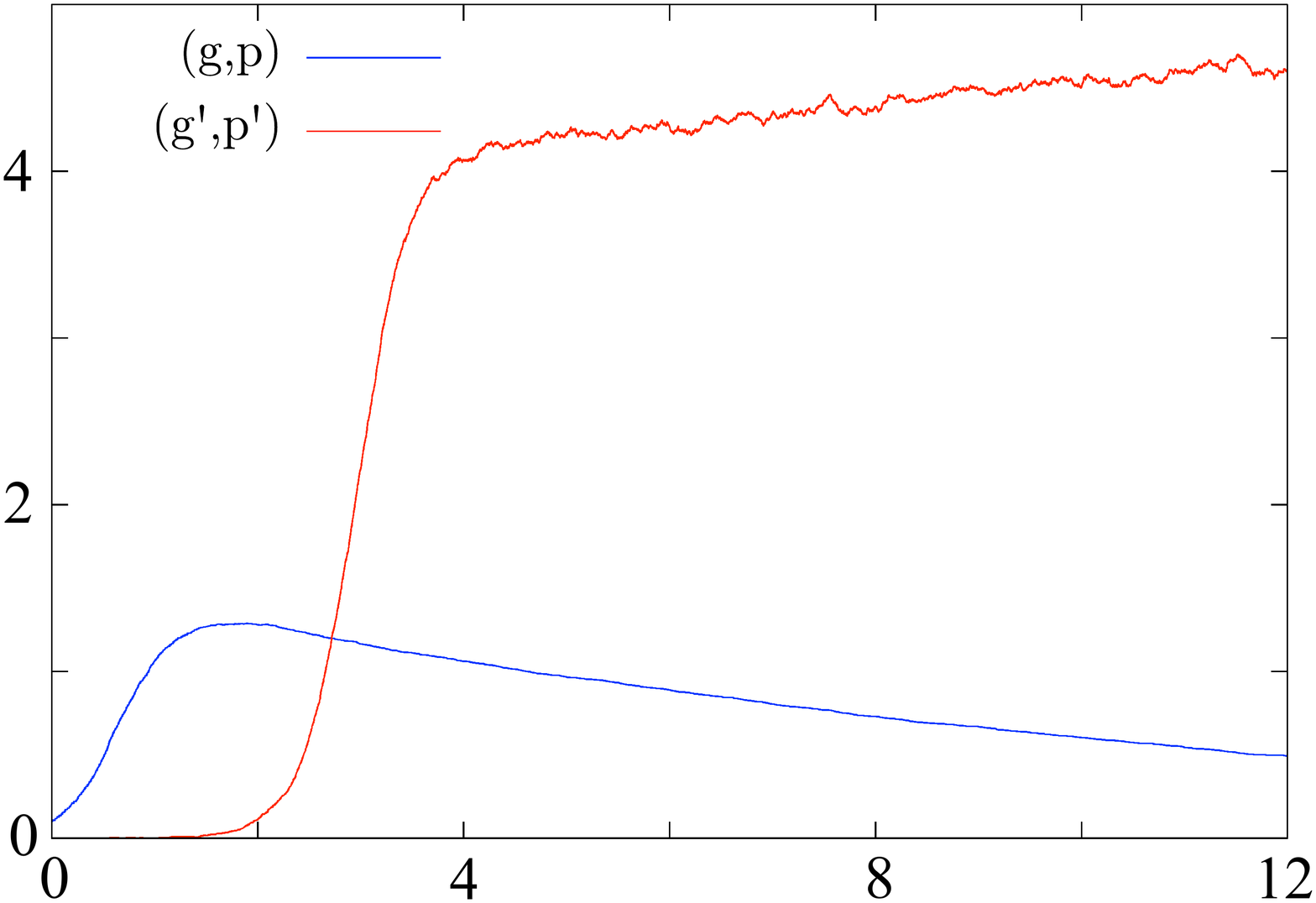}
		\end{minipage}\hfill
		\begin{minipage}{.49\textwidth} 
			B \includegraphics[width=\textwidth]{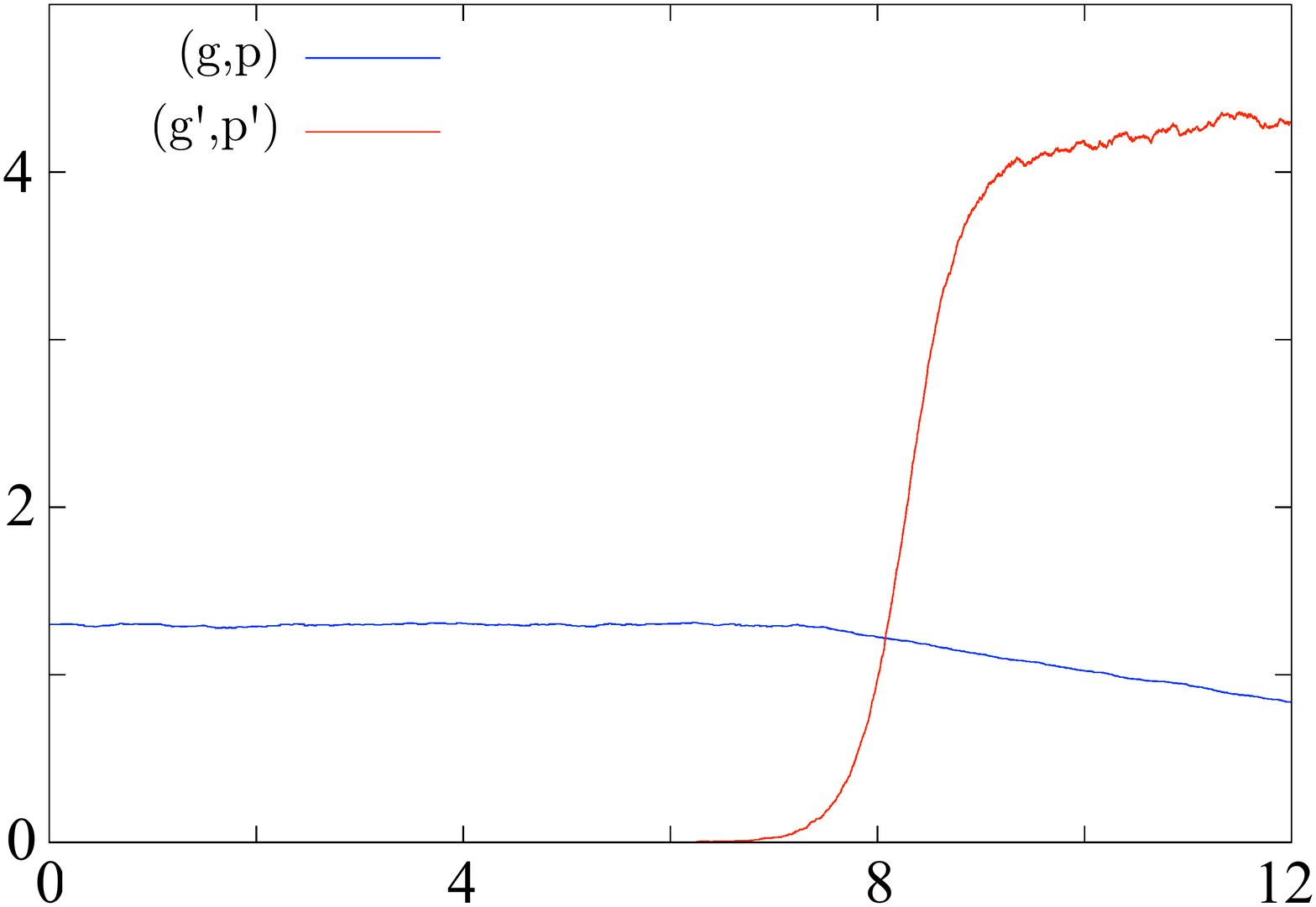}
		\end{minipage}\hfill
		\begin{minipage}{.49\textwidth} 
			C \includegraphics[width=\textwidth]{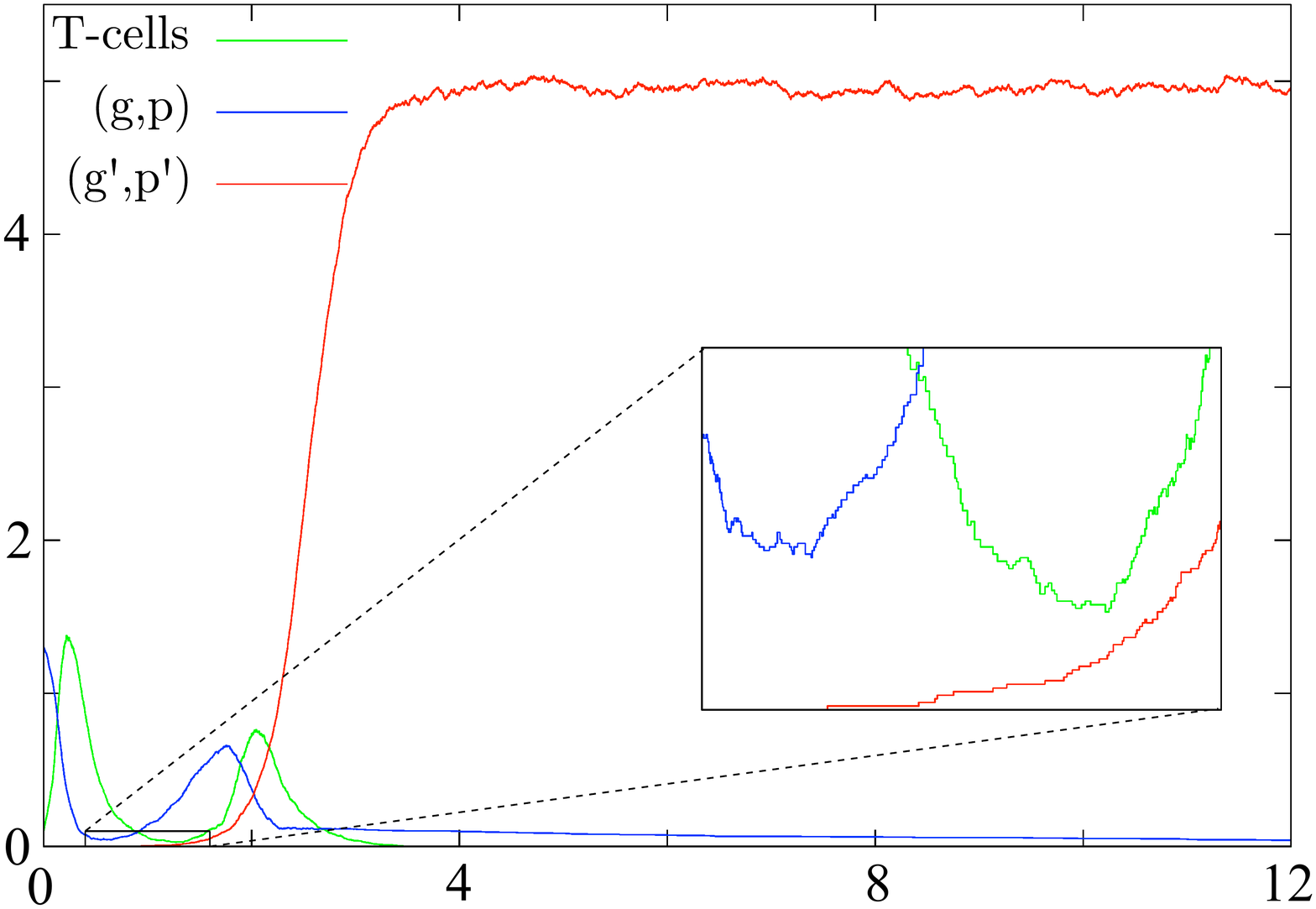}
		\end{minipage}\hfill
		\begin{minipage}{.49\textwidth} 
			D \includegraphics[width=\textwidth]{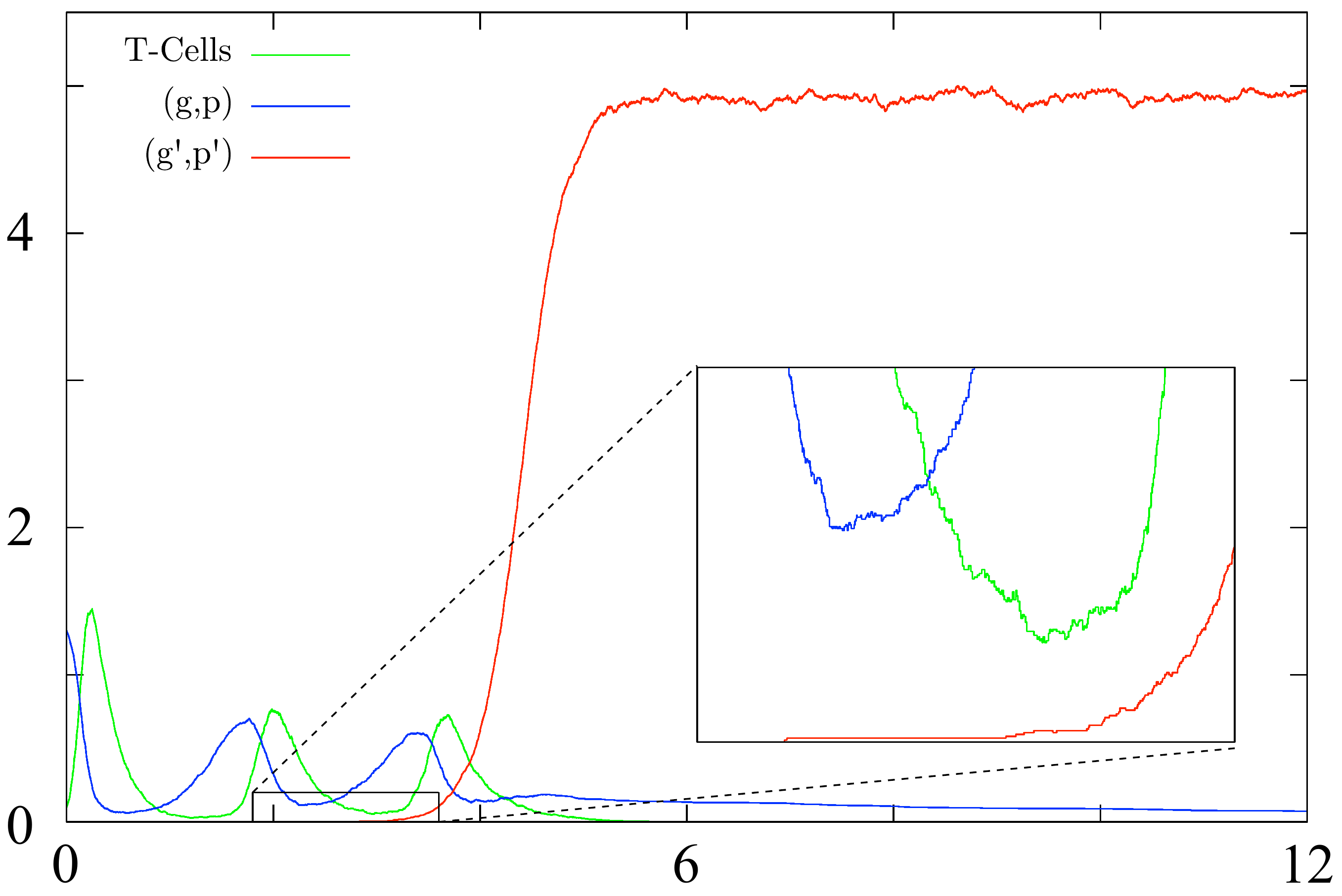}
		\end{minipage}
	\caption{
	{Simulations of mutation events in a population, where competition is acting via birth reduction (parameters are given by \eqref{IMT-parameter}). The number of individuals divided by $1000$ is plotted versus time: 
		Effect for an initial population which is small (A), 
		or at equilibrium (B) 
		or under therapy (C) and (D).
		\label{IMT}}
	}
\end{figure}

\section{Discussion}
Therapy resistance is a major issue in the treatment of advanced stages of cancer. We have proposed a stochastic 
mathematical model that allows to simulate treatment scenarios and applied it to the specific case of immunotherapy of 
melanomas. Comparison to experimental data is so far promising. The models pose challenging new mathematical 
questions, in particular due to the interplay of fast phenotypic switches and rare driver mutations. First numerical results 
point to a significant effect of stochastic fluctuations in the success of therapies. More precise experimental data will be 
needed in the future to fit crucial model parameters. 
While our models describe the actions of individual cells and cytokines, they do not by far resolve the full complexity of 
the biological system. In particular, they do not reflect the spatial structure of the tumour and its microenvironment.  Also, 
the distinction of only two phenotypes of the tumour cells is a simplification. The same is true  for the interaction with other immune cells and cytokines. This reflects on the one hand the limitation due to available experimental data, on the other hand the use of
a model of reduced complexity also makes numerical computations and theoretical understanding of the key phenomena 
feasible. The rates entering as model parameters therefore have to be understood as \emph{effective} parameters, e.g.\ the death rate of T-cells accounts for their natural death as well as the exhaustion phenomenon. In principle it is possible to 
increase the resolution of the model; this, however, increases the number of 
parameters that need to be determined experimentally which would pose a major challenge. Already at the present stage, 
the model parameters are not known well enough and are adjusted to reproduce the experimental findings. Some 
parameters that it would be very useful to see measured precisely are: 
\begin{itemize}
\item birth and death rates of tumour cells, both in
differentiated and dedifferentiated types. Currently these are estimated from the growth rate of the tumour, but this 
yields only the difference of these rates;
\item killing rates of T-cells, both of the differentiated and the dedifferentiated tumour cells;
\item rates of phenotypic switches, both in the absence and the presence of TNF-$\alpha$;
\item death rates of T-cells and their expansion rates when interacting with tumour cells.
\end{itemize}

\noindent
Nevertheless, we see the proposed model as a promising tool to assist the development of improved treatment protocols. Simulations may guide the choice of experiments such that  the number of necessary experiments can be reduced.
The obvious strength of our approach is to model reciprocal interactions and phenotypic state transitions of tumour and immune cells in a heterogeneous and dynamic microenvironment in the context of therapeutic perturbations.

The clinical importance of phenotypic coevolution in response to therapy has been recently documented in patients' samples from melanomas acquiring resistance to MAPK inhibitors \cite{Hugo}. Adaptive activation of bypass survival pathways in melanoma cells was accompanied by the induction of an exhausted phenotype of cytolytic T cells. This has important implications for the combinatorial use of cancer immunotherapy (checkpoint inhibitors like anti PD-1) with respect to scheduling. We envision that our mathematical approach will help to integrate such patient omics data with experimental findings to guide novel strategies. Of note and similar to our previous study, dedifferentiation of melanoma cells was identified as a major mechanism of escape from MAPK inhibitors \cite{Muller, Koni}. We recently dissected the molecular circuitries that control melanoma cell states and showed how melanoma dedifferentiation governs the composition of the immune cell compartment through a cytokine-based crosstalk in the microenvironment \cite{Ries, Holz}. Hence, malignant melanoma is a paradigm for a phenotypic heterogeneous tumour and a future goal is to incorporate this increasing knowledge of melanoma cell plasticity into our method to refine its capability to model complex interactions with immune cells.

Importantly, phenotypic plasticity in response to therapy is a widespread phenomenon and non-small cell lung cancer (NSCLC) represents a prominent example. A subset of NSCLCs harbour activating mutations in the epidermal-growth factor receptor EGFR and respective small molecule inhibitors (EGFRi) are potent first line cancer drugs for this NSCLC subtype. However, tumours invariably relapse and genetic selection of subclones with the second site resistance mutation EGFRT790M is the major event. A substantial number of relapse tumours show remarkable transitions from an NSCLC adenocarcinoma to a neuroendocrine-related small cell lung cancer (SCLC) phenotype \cite{Seq}. Given the recent success of immune checkpoint inhibitors in NSCLC \cite{Bra}, it will be of clinical interest to investigate the phenotypic coevolution of immune cells in the context of NSCLC-SCLC lineage transitions. Again, our mathematical approach could represent a valuable tool to support this research. Finally, our results suggest that stochastic events play an unanticipated critical role in the dynamic evolution of tumours and the emergence of therapy resistance that requires further experimental and clinical investigation.

\newpage

\appendix\section{Pseudo-code}\label{appendix}

The depth diagram of the algorithm we used to generate the simulations in this article is given in Figure \ref{depth-diagram}.
The pseudo-code  is given in Algorithm \eqref{algo}.\\

We use the following notations:
let $\calD$ be some discrete set and $X$ a $\calD$-valued random variable, then $X$ sampled according to the weights $\{w(i),i\in\calD\}$ means that $\bbP(X=i)=w(i)/\sum_{i\in\calD}w(i)$.

\begin{figure}[h!]
	\centering
	\includegraphics[width=\textwidth]{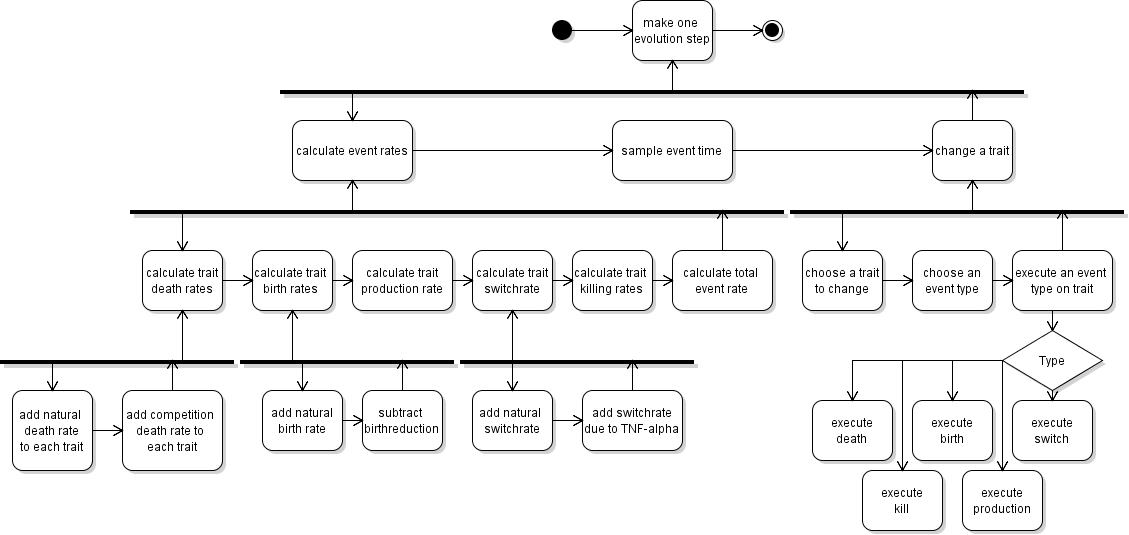}
	\caption
	{Depth diagram of the program.}
	\label{depth-diagram}
\end{figure}

\begin{algorithm}[H]
	\footnotesize
	\DontPrintSemicolon
	\KwData{initial conditions: $\nu^K_0\in\calM^K(\calX)$,  number of iterations: $N_{\rm    Iterations}$,
		parameters described in Section 1 }
	$T_0\leftarrow0$,
	$\nu^K_{T_0}\leftarrow \nu^K_0$, $k\leftarrow0$\;
	\While{$k\leq N_{\rm Iterations}$}{
		\For{$x\in \text{Supp}(\nu^K_{T_k})$}{
			\If{$x={(g,p)}\in\calG\times\calP$}{
				$B(x) \leftarrow K\nu^{K}_{T_k}(g,p)\:\bigg\lfloor b(p)-\sum_{(\tilde g, \tilde
					p)\in\text{Supp}(\nu^K_{T_k})}c_b(p,\tilde p)\nu^K_{T_k}(\tilde g,\tilde
				p)\bigg\rfloor_+$\\
				$D(x) \leftarrow K\nu^{K}_{T_k}(g,p)\:\bigg(
				d(p)+
				\bigg\lfloor b(p)-\sum_{(\tilde g,
					\tilde p)\in\text{Supp}(\nu^K_{T_k})}c_b(p,\tilde
				p)\nu^K_{T_k}(\tilde g,\tilde p)\bigg\rfloor_-$\\\hspace{6cm}$+
				\sum_{(\tilde g, \tilde p)\in\text{Supp}(\nu^K_{T_k})}
				c(p,\tilde p) \nu^K_{T_k} (\tilde g,\tilde p)
				\bigg),$\\[0.1em]
				$T(x) \leftarrow  K\nu^K_{T_k}(g,p)\:\sum_{z\in\text{Supp}(\nu^K_{T_k})}  t(z,p)\nu^{K}_{T_k}(z)$,\quad
				$S(x) \leftarrow K\nu^{K}_{T_k}(g,p)\:\sum_{\tilde p\in\calP} \left(s^g(p,\tilde p)+
				\sum_{w\in\text{Supp}(\nu^K_{T_k})}s_w^g(p,\tilde p) \nu^K_{T_k}(w)\right)$,\\
				$P(x)\leftarrow 0,$
			}
			\If{$x={z}\in\calZ$}{
				$B(x) \leftarrow K\nu^{K}_{T_k}(z)\:b(z),$\quad
				$D(x) \leftarrow K\nu^{K}_{T_k}(z)\:d(z),$\quad $T(x)\leftarrow 0$,\quad
				$S(x) \leftarrow 0$
				$P(x) \leftarrow K\nu^{K}_{T_k}(z) \:\sum_{(g,p)\in\text{Supp}(\nu^K_{T_k})} b(z,p)\nu^K_{T_k}(g,p)$,\quad 
			}
			\If{$x={w}\in\calW$}{ 
				$B(x) \leftarrow 0, $\quad$D(x) \leftarrow K \nu^{K}_{T_k}(w)\:d(w),$\quad 
				$T(x) \leftarrow 0,$\quad $S(x) \leftarrow 0$,\quad
				$P(x) \leftarrow 0,$ \quad 
				
			}
			$\text{TotalTraitRate}(x) \leftarrow B(x)+D(x)+T(x)+P(x)+S(x)$
		}
		TotalEventRate $\leftarrow\sum_{x\in \text{Supp}(\nu^K_{T_k})}$TotalTraitRate$(x)$\;
		Sample $t\sim\calE\rm{xp}(\rm TotalEventRate)$\;
		$T_{k+1}\leftarrow T_k+t$\;
		Sample $ x_{\text{chosen}}\in\calX$ according to the weights  $\left\{\text{TotalTraitRate}(x), x\in \text{Supp}(\nu^K_{T_k}))\right\}$.    \;
		Sample the event $E\in\{\text{Birth}, \text{Death},\text{Therapy},\text{Production},\text{Switch}\}$  according to the weights 
		$\{B(x_{\text{chosen}}),D(x_{\text{chosen}}),T(x_{\text{chosen}}),P(x_{\text{chosen}}), S(x_{\text{chosen}})\}$\;
		\Case{$E= \textbf{Birth}$}{
			\If{$x_{\text{chosen}}={(g,p)}\in\calG\times\calP$}{
				Sample $ {B}\in\{\text{No Mutation},\text{Mutation}\}$ according to  $\{1-\mu_g,\mu_g\}$\;
				\Case{${B}=\textbf{No Mutation}$}{
					$ \nu^K_{T_{k+1}}\leftarrow \nu^K_{T_k}+\frac 1 K \delta_{x_{\text{chosen}}} $ }
				\Case{${B}= \textbf{Mutation}$}{
					Sample $(\tilde g, \tilde p)$ according to $m((g,p),(\tilde g,\tilde p))$\;
					$ \nu^K_{T_{k+1}}\leftarrow \nu^K_{T_k}+\frac 1 K \delta_{(\tilde g, \tilde p)} $
				}}
				\Else{ $ \nu^K_{T_{k+1}}\leftarrow \nu^K_{T_k}+\frac 1 K \delta_{x_{\text{chosen}}} $ }}
			\Case{$E=$ \textbf{Death}}{
				$\nu^K_{T_{k+1}}\leftarrow \nu^K_{T_k}-\frac 1 K \delta_{x_{\text{chosen}}} $ }
			\Case{$E=$ \textbf{Therapy}} {Note that $x_{\text{chosen}}=(g,p)$ for some $(g,p)\in\calG\times\calP$ in this case.\;
				Sample $z\in\calZ$ according to the weights 
				$\left\{ t(z,p)\nu^K_{T_k}(z) ,z\in\text{Supp}(\nu^K_{T_k})\cap\calZ\right\}$\\[0.2em]
				
				$\nu^K_{T_{k+1}}\leftarrow \nu^K_{T_k}-\frac 1 K \delta_{(g,p)} +\sum_{w\in\calW}\ell^{\text{kill}}_w(z,p)\frac 1K  \delta_{w}  $ 
			}
			\Case{$E=$ \textbf{Production}}{          
				Sample $(g,p)\in\text{Supp}(\nu^K_{T_k})$ according to the weights $\left\{{ b(x_{\text{chosen}},p)\nu^K_{T_k}(g,p) },(g,p)\in\text{Supp}(\nu^K_{T_k})\right\}$\;
				$\nu^K_{T_{k+1}}\leftarrow \nu^K_{T_k}+\frac 1 K \delta_{x_{\text{chosen}}} +\sum_{w\in\calW}\ell^{\text{prod}}_w(x_{\text{chosen}},p)\frac 1K  \delta_{w}  $ 
			}
			\Case{$E=$ \textbf{Switch}} 
			{ Note that $x_{\text{chosen}}=(g,p)$ for some $(g,p)\in\calG\times\calP$ in this case.\;
				Sample $\tilde p\in\calP$ according to $\left\{ {s^g(p,\tilde p)+ \sum_{w\in\text{Supp}(\nu^K_{T_k})}
					s_w^g(p,\tilde p) \nu^K_{T_k}(w)},\tilde p\in\calP\right\}$ \;
				$\nu^K_{T_{k+1}}\leftarrow \nu^K_{T_k}-\frac 1 K \delta_{(g,p)} +\frac 1 K \delta_{(g,\tilde p)} $
			}   
			$ k=k+1$}
		\caption{\small
			Pseudo-code of the Gillespie  algorithm used for generating
			the figures. }
		\label{algo}
	\end{algorithm}


\newpage

\bibliography{bibliography}
%
%
\end{document}